\documentclass[journal]{IEEEtran}
\usepackage[cmex10]{amsmath}
\usepackage{amsmath,amssymb}
\usepackage{subfigure}
\usepackage{graphicx,graphics,color,psfrag}
\usepackage{cite,balance}
\usepackage{caption}
\allowdisplaybreaks
\usepackage{accents}
\usepackage{amsthm}
\usepackage{bm}
\usepackage[english]{babel}
\usepackage{multirow}
\usepackage{enumerate}
\usepackage{cases}
\usepackage{stfloats}
\usepackage{dsfont}
\usepackage{color,soul}
\usepackage{amsfonts}
\usepackage{cite,graphicx,amsmath,amssymb}
\usepackage{subfigure}
\usepackage{fancyhdr}
\usepackage{hhline}
\usepackage{graphicx}
\usepackage{array,color}
\usepackage{booktabs}
\usepackage{diagbox}
\usepackage{indentfirst}
\usepackage{url}
\usepackage[ruled,vlined]{algorithm2e}

\usepackage{amsmath}

\usepackage{lipsum}

\newtheorem{assumption}{Assumption}
\newtheorem{lem}{Lemma}

\theoremstyle{definition}
\newtheorem{defn}{Definition}

\newtheorem{rem}{Remark}
\newtheorem{theo}{Theorem}
\usepackage{lipsum}

\begin{document}
	\title{Turning Channel Noise into an Accelerator for Over-the-Air Principal Component Analysis}
	
	\author{{Zezhong~Zhang,~\IEEEmembership{Member,~IEEE,} Guangxu~Zhu,~\IEEEmembership{Member,~IEEE,} Rui~Wang,~\IEEEmembership{Member,~IEEE,} Vincent~K.~N.~Lau,~\IEEEmembership{Fellow,~IEEE,} and~Kaibin~Huang,~\IEEEmembership{Fellow,~IEEE,}
	
	\thanks{Manuscript received April 20, 2021; revised July 17, 2021 and December 2, 2021, and accepted March 21, 2022. The work described in this paper was substantially supported by a fellowship award from the Research Grants Council of the Hong Kong Special Administrative Region, China (Project No. HKU RFS2122-7S04).
	The work was also supported by Guang-dong Basic and Applied Basic Research Foundation under Grant 2019B1515130003, Hong Kong Research Grants Council under Grants 17208319 and 17209917, Innovation and Technology Fund under Grant GHP/016/18GD, Shenzhen Science and Technology Program under Grant JCYJ20200109141414409, and National Natural Science Foundation of China under grant 62171213. The work of Guangxu Zhu was supported in part by the National Key Research and Development Program of China under Grant 2018YFB1800800, in part by the National Natural Science Foundation of China under Grant 62001310.	
	The associate editor coordinating the review of this paper and approving it for publication was O. Tirkkonen. \emph{(Corresponding authors: Kaibin Huang and Rui Wang.)}}}
	
	\thanks{Z. Zhang and K. Huang are with The  University of  Hong Kong, Hong Kong (Email: \{zzzhang, huangkb\}@eee.hku.hk).} 
	\thanks{G. Zhu is with Shenzhen Research Institute of Big Data, China (Email: gxzhu@sribd.cn).} 
	\thanks{R. Wang is with  Southern University of Science and Technology, China (Email: wang.r@sustech.edu.cn).} 
	\thanks{V. K. N. Lau is with the The Hong Kong University of Science and Technology, Hong Kong (Email: eeknlau@ust.hk).}
	}
	\date{}
	\maketitle

\begin{abstract}

The enormous data distributed at the network edge  and ubiquitous connectivity have led to the emergence of the new paradigm   of distributed machine learning and large-scale data analytics. Distributed \emph{principal component analysis} (PCA) concerns finding a low-dimensional subspace that contains the most important information of  high-dimensional data distributed over the network edge. The subspace is useful for distributed data compression and feature extraction. This work advocates the application of over-the-air  federated learning to efficient implementation of   distributed PCA in a wireless network under a data-privacy constraint, termed \emph{AirPCA}.  The design features the exploitation of the waveform-superposition property of a multi-access channel to  realize \emph{over-the-air aggregation} of local subspace updates computed and simultaneously  transmitted by devices to a server, thereby reducing the multi-access latency. The original  drawback of this class of techniques, namely channel-noise perturbation to uncoded analog modulated signals, is turned into a mechanism for escaping from saddle points during  \emph{stochastic gradient descent} (SGD) in the AirPCA algorithm. As a result, the convergence of the AirPCA algorithm is accelerated. To materialize the idea, descent speeds in different types of descent regions are analyzed mathematically using martingale theory by accounting for wireless propagation and techniques including broadband transmission, over-the-air aggregation, channel fading and noise. The results reveal the accelerating  effect of noise in saddle regions and the opposite effect in other types of regions. The insight and results are applied to designing   an online scheme for adapting receive signal power to the type of current descent region. Specifically, the scheme amplifies the noise effect in saddle regions by reducing signal power and applies the power savings to suppressing the  effect  in other regions. From experiments using real datasets, such power control is found to accelerate convergence while achieving the same  convergence accuracy as in the ideal case of centralized PCA.

\end{abstract}

	\section{Introduction}
	
The enormous data distributed over edge  devices (e.g., smartphones and Internet-of-Things sensors) and ubiquitous connectivity have triggered the paradigm shift of distributed machine learning and large-scale data analytics \cite{FLSurvey}. As a standard technique in data analytics, \emph{principal component analysis} (PCA) provides a simple way of discovering a low-dimensional subspace, called \emph{principal components}, that contains the most important information of a high-dimensional  dataset \cite{Tut1}. This is  useful for  data compression, simplification of data description, and feature extraction. For these reasons, PCA finds  applications in almost all scientific fields ranging from wireless communication (see e.g.,  \cite{CPCA2,CPCA3}) to machine learning (see e.g., \cite{CPCA1,Eigenface}). A common approach of PCA is based on  \emph{singular-value decomposition} (SVD) of a data table, which comprises all data samples as rows. However, the required data centralization makes  this    approach   infeasible  for implementing PCA  in a mobile network  as uploading mobile data  violates their privacy and changes their ownerships. Addressing the issue has motivated researchers to apply \emph{federated learning} (FL) that preserves  data privacy to implementing distributed PCA, termed \emph{federated PCA} \cite{FPCA}. Federated PCA, or distributed PCA, can help compress and simplify the data distributed at the network edge, e.g., data generated by vehicular sensing or AR/VR applications and collected by different devices, for convenient storage and their further use in edge learning. As originally proposed for distributed learning, the FL framework involves devices in updating a prediction model using local data and uploading local updates (instead of data) to a server for aggregation to update the global model \cite{Decentralized}. In this way, the ``data privacy'' is preserved as elaborated in \cite{Decentralized} -- \emph{``Federated learning protects data ownership of devices by avoiding uploading raw data while providing a mechanism to leverage distributed mobile data. Specifically, a server requests each device to upload updates on the global model as computed using local training data. In general, the updates do not directly expose the content of local data and contain  much less information than the latter, thereby protect the users' data ownership.''} In this work, we propose an  efficient design of federated  PCA in a wireless system  based on \emph{over-the-air federated learning} which exploits the waveform-superposition property of a multi-access channel to realize low-latency over-the-air aggregation  \cite{BB,GX}. Targeting this design named  \emph{over-the-air PCA} (AirPCA), a power-control scheme is proposed to  adapt transmission power of devices to \emph{stochastic gradient descent} (SGD) such that channel noise is turned into an accelerator for the descent. 

As originally proposed in \cite{Iwen}, federated PCA  involves devices in computing their estimates of principal components via  SVD of their local data and uploading  their local estimates to a server for aggregation to obtain the global estimate, called as the one-shot method. There is a key drawback of the one-shot method that sharing the local principal components concerns partial data privacy. On the other hand, uploading full-SVD results leads to high communication latency when the number of devices grows large. By moderately reducing the dimensionality of the local subspace estimates, the communication latency issue is alleviated \cite{FPCA,DisEst,improvedPCA}. However, the dimension reduction on local subspace estimates results in a biased error, which distorts the global estimate when the local datasets are highly non-independent identically distributed. Another solution to federated PCA is to apply the well-known power method, which can be integrated with over-the-air aggregation to provide fast convergence and negligible communication latency \cite{OTAPowerMethod}. However, the power method is sensitive to the noise perturbation, making it infeasible in a wireless network, especially when the SNR is low. In view of the drawbacks of the existing methods, in this work we seek to apply SGD-based algorithms to solving federated PCA as an optimization problem of finding a subspace (principal components) to minimize the error function for data compression by projection onto the subspace. The above idea follows from the Oja's method \cite{Oja}, which solves centralized PCA using SGD-based algorithms. In the context of federated PCA, the main difficulty  for applying SGD  arises from the unitary/orthogonal  constraint of the optimization variable that is a subspace, which makes the optimization problem indecomposable. As elaborated in \cite{Decentralized}, FL cannot be directly applied for indecomposable optimization problems. The difficulty  can be overcome using the finding in \cite{SubspaceTracking} that the solution to the  unconstrained problem without the unitary/orthogonal constraint also solves the original constrained problem. In this work, we show that the SGD method is robust against channel noise. Moreover, with the presence of channel noise, we prove that the SGD method guarantees convergence to the global optimum through both analysis and simulations, which thus beats the power method. Moreover, by adopting over-the-air aggregation in the gradient uploading phase, the communication latency issue is also addressed, making the SGD algorithm outperform one-shot methods in \cite{FPCA,DisEst,improvedPCA} when the number of devices is large.

In a scenario with many devices and high-dimensional data, the uploading of local model updates from devices can cause a communication bottleneck for FL (including federated PCA) \cite{Decentralized}. Overcoming the bottleneck is a main research theme for  FL in wireless networks. A wide range of relevant techniques have emerged recently ranging from source encoding \cite{Mingzhe,Quek} to resource management \cite{Yuqing,Mingzhe2}, where energy efficient FL also attracts much attention \cite{EnergyEfficient1,EnergyEfficient2,QSEfficiency,EnergyEfficient3}. In particular, the mentioned over-the-air FL is a class of techniques that realize over-the-air aggregation by superimposing analog modulated model updates transmitted simultaneously by devices \cite{BB, Kwan,MeixiaTao,Bennis,Deniz,yuanmingShi,Dongzhu}. Compared with digital orthogonal access, over-the-air aggregation supporting simultaneous access has the advantage of reducing  the multi-access latency when the number of devices is large \cite{BB}. However, the uncoded analog transmission exposes the receive signals to the perturbation of channel noise that can potentially degrade the learning performance. In this work, we make an attempt on  turning  the drawback into an advantage in the context of AirPCA by exploiting the characteristics of the mentioned error function for AirPCA, which are described as follows. For training a model (e.g., a deep neural network) using  FL, the (prediction) loss function is dataset dependent and has no known expression. On the contrary, the PCA error function is well defined and its theoretical properties are well understood in the literature. To be specific, the error function has a finite number of stationary points comprising  a global optimum and a number of discrete saddle points \cite{SubspaceTracking}.  Consequently, the regions along the descent path belong to one of the three types: 1) a \emph{saddle region} centered at  an associated saddle point, 2) a \emph{non-stationary region} with relatively large slopes, and 3) an \emph{optimum region} centered at the global optimal point (see illustrations in Fig. \ref{ThreeRegions} in the sequel). The properties suggest that the gradient descent  can be trapped at a saddle point having a zero gradient  if the descent path encounters a saddle region. The problem is well known and a common solution is to  add artificial noise to gradients to escape from saddle points \cite{Ge15}. On the other hand, the noise slows down the descent outside saddle regions and reduces  the convergence accuracy. Instead of adding artificial noise,  we propose the idea of leveraging channel noise existing in received signals in AirPCA to help escape from saddle points by amplifying its effect but reducing its effect in other types of regions on the descent path. 

The idea is materialized in this work by designing region-adaptive power control for AirPCA. The main contributions  are summarized as follows.

\begin{itemize}
		\item {\bf Descent-Speed Analysis: }  Building on the martingale-based analytical approach for centralized PCA training in~\cite{Ge15}, we develop a new framework of  descent-speed analysis  for AirPCA. In light of prior work on distributed PCA assuming reliable links,  the novelty of the framework  lies in accounting  for wireless propagation and techniques, including \emph{orthogonal frequency division multiplexing} (OFDM), over-the-air aggregation, channel fading and noise.  The descent speed of AirPCA is measured by the reduction of the expected error function over a given number of communication rounds. Using the framework and exploiting the mentioned properties of the error  function,  the descent speeds in different regions on a descent path are characterized mathematically. Consider the gradient descent in a non-stationary region. A lower bound on the descent speed is derived and shown to be a monotone \emph{increasing} function of the expected receive \emph{signal-to-noise ratio} (SNR), which is uniform for all devices as a result of signal-magnitude alignment in over-the-air aggregation, and also the expected number of active devices in the presence of fading. In contrast, the descent speed in a saddle region is a monotone \emph{decreasing} function of these two variables as their reduction amplifies the noise effect and accelerates the escape from the saddle point. Last, it is proved that under the effect of channel noise, the descent path can eventually enter the optimum region in probability so long as the step-size is sufficiently small.

\item {\bf Region-Adaptive Power Control:} Based on the analytical result, a simple scheme for online power control is designed to adapt the uniform receive SNR to the type of current descent region by coordinating transmission power of devices. Thereby,  the gradient descent of AirPCA is accelerated. To be specific, when a saddle region is detected, the receive SNR is fixed at a minimum value to amplify the noise effect so that the descent path can escape from the saddle point. This results in power savings under an average power constraint. On the other hand, when either a non-stationary or the optimum region is detected, receive SNR is enhanced by either using up all power savings from preceding rounds in the current round, called \emph{one-shot saving spending}, or distributing  the savings over multiple  rounds using a diminishing geometric sequence with the common ratio  controlling the saving-dissipation speed, called \emph{gradual saving spending}.

\item {\bf Experimental Results:} The learning  performance of AirPCA is evaluated using experiments with several well-known real datasets, namely MNIST, CIFAR-10, and AR. The proposed region-adaptive power control is shown to be effective in escaping from saddle points and  accelerating AirPCA convergence with respect to the case with fixed receive power/SNR. At the same time, the proposed scheme achieves the convergence accuracy of centralized PCA. Moreover, it is found that the mentioned gradual saving spending can outperform the one-shot counterpart if the common ratio is optimized. The effects of other parameters such as the number of devices and the channel-truncation threshold are also investigated. 
	\end{itemize}
	
The reminder of the paper is organized as follows. The AirPCA system  is described  in Section \ref{SystemModel}. In Section \ref{convergenceAna}, descent speeds of  AirPCA are  analyzed. Based on the analytical results, the online scheme of region-adaptive power control  is designed  in Section \ref{PowerControl}. Experimental results are presented in Section \ref{Sim}, followed by concluding remarks in Section \ref{Conclusion}.
\begin{figure*}[t]
	\centering
	\includegraphics[height=120pt]{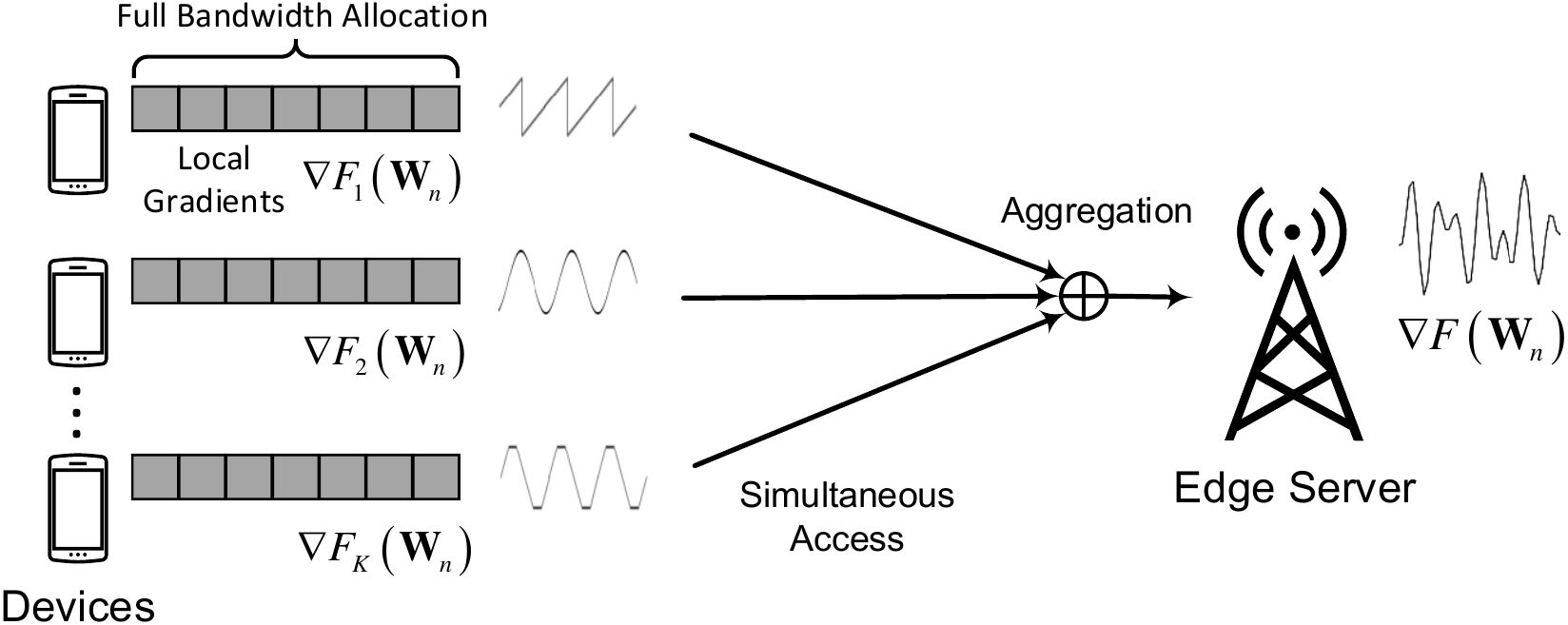} \caption{Broadband AirPCA system.}\vspace{-0.3cm}
\end{figure*}	
	
	\section{Over-the-Air PCA System}\label{SystemModel}
	In this section, we present the model of a broadband AirComp system, design the distributed PCA algorithm, and describe its implementation in the system.

	\subsection{Over-the-Air Aggregation  System}\label{broadband}
	We consider the  broadband over-the-air aggregation  system as proposed in \cite{BB} to support  AirPCA. In the system, there are $K$ devices communicating with a single server. The communication comprises multiple rounds,  each of which is divided into an  uplink and a  downlink transmission phases.  Consider the uplink phase of an arbitrary round. Each device transmits a fixed number, denoted as $c$, of symbols to the server over $M$ (frequency) sub-channels generated by OFDM. To this end, $c$ symbols are divided into $\frac{c}{M}$ blocks. Each block is transmitted in one OFDM symbol duration with  each sub-channel modulated with one symbol using linear analog modulation. The transmission of all devices is simultaneous so as to realize over-the-air aggregation. Then the $i$-th aggregated symbol received by the server  in the $n$-th communication round, denoted as $y^{(i)}_n$, is given as 
	\begin{equation}\label{Eq:RxSymb}
	y^{(i)}_n = \sum\limits_{k=1}^{K} h_{k,n}^{(i)}p_{k,n}^{(i)}s_{k,n}^{(i)}+z^{(i)}_n, \qquad 1\le i\le M, n \ge 1,
	\end{equation}
	where $s_{k,n}^{(i)}$ denotes the symbol transmitted by device $k$ with $\mathsf{E}[|{s}_k^{(i)}|^2]=1$,  the Gaussian random variables $h_{k,n}^{(i)}\sim\mathcal{CN}(0,1)$ and $z^{(i)}_n\sim\mathcal{CN}(0,\sigma^2)$ represent the  gain and noise of the corresponding sub-channel, respectively, and $p_{k,n}^{(i)}$ is the precoding  coefficient. Let $P_{k, n}$ denote the power consumption by the broadband transmission of device $k$ in round $n$: $P_{k, n} =\sum\limits_{i=1}^M |p_{k,n}^{(i)}|^2$.  The transmission of each device is subject to an average power constraint:  
	\begin{equation}\label{powerConstraint}
	\mathsf{E}[P_{k, n}] = \mathsf{E}\left[ \sum\limits_{i=1}^M |p_{k,n}^{(i)}|^2\right]  \le \bar{P},
	\end{equation}
	for a given constant $\bar{P}$.  
	
	\begin{figure*}[t]
		\begin{minipage}[c]{1\linewidth}
			\centering
			\begin{center}
				\includegraphics[height=64pt]{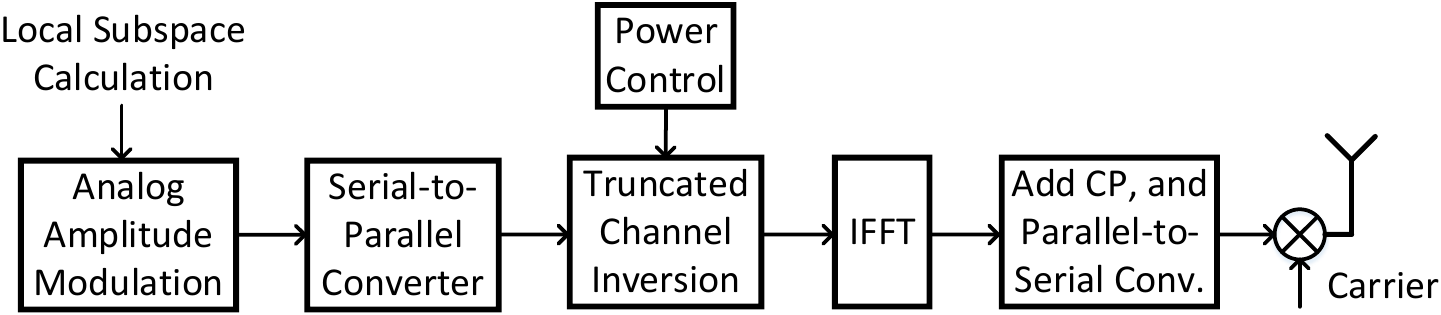}
			\end{center}			
			\caption*{(a) Transmitter design for edge devices.}
		\end{minipage}
		\begin{minipage}[c]{1\linewidth}
			\centering
			\begin{center}
				\includegraphics[height=69pt]{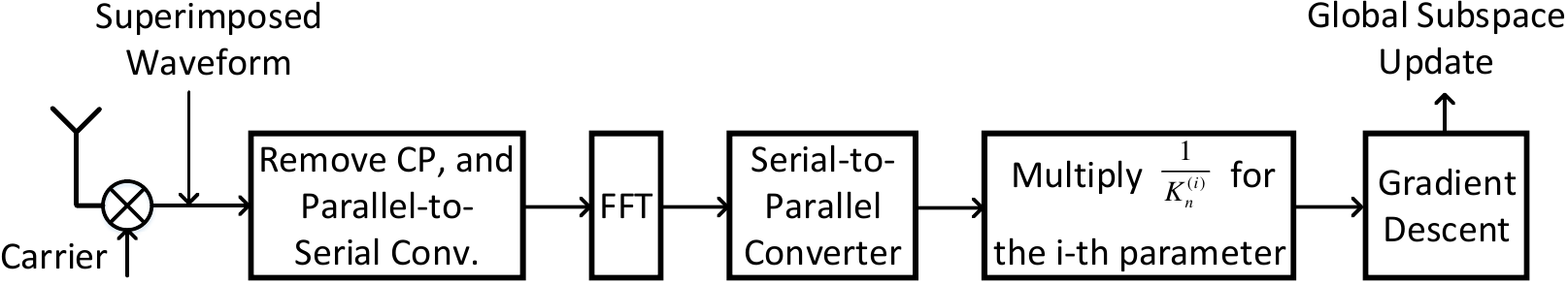}
			\end{center}
			\caption*{(b) Receiver design for the edge server.}
		\end{minipage}
		\caption{Transceiver design of the AirPCA system.}\vspace{-0.3cm}
	\end{figure*}

	Over-the-air aggregation requires channel inversion so that each received symbol is the desired sum of transmitted symbols. We adopt one existing scheme that is designed to satisfy  the average power constraint called  \emph{truncated channel inversion} \cite{BB,1Bit}. Specifically, the precoding  coefficient in \eqref{Eq:RxSymb} is given as 
	\begin{align}\label{truncation}
	p_{k,n}^{(i)} = \left\{ \begin{aligned}
	&\frac{{\sqrt {P^{\text{rx}}_n} }}{{h_{k,n}^{(i)}}},& {\left| {h_{k,n}^{(i)}} \right|^2} \geq {G},  \\
	&0,& \ \ {\left| {h_{k,n}^{(i)}} \right|^2} < {G}, \\ 
	\end{aligned}  \right.
	\end{align}
	where the controllable receive power $P^{\text{rx}}_n$ and constant $G$ are called signal-magnitude-alignment factor and truncation threshold, respectively, as explained in the following. The factor $P^{\text{rx}}_n$, which scales magnitude of an aggregated symbol at the receiver, forms a \emph{power-control sequence $\{P^{\text{rx}}_n\}$} in the entire process  controlling the receive power under the constraint in \eqref{powerConstraint}. Given identical distributions of sub-channel gains, it can be obtained that~\cite{BB}
	\begin{align}\label{powerResult}
	\mathsf{E} [P^{\text{rx}}_n] \le \frac{\bar{P}}{M \mathsf{Ei}(G)}=\bar{P}^{\text{rx}}_\text{max},
	\end{align}
	where $\mathsf{Ei}(G) \!\!\triangleq\!\! \int_G^\infty \frac{1}{t}\exp(-t)dt$ is the exponential integral function. On the other hand, the truncation threshold $G$ avoids excessive power consumption due to inversion of deeply faded sub-channels. To enforce fixed transmission latency, the symbols assigned to truncated sub-channels are discarded. The probability that a sub-channel avoids truncation  (or equivalently  its symbol is transmitted) is called  \emph{activation probability} and denoted by $\zeta^\text{act}$. It is easily  obtained as 
	\begin{align}\label{nonTruncation}
	\zeta^\text{act} = \Pr(|h_{k,n}^{(i)}|^2 \ge G) = e^{-G}.
	\end{align}
	The value  $\zeta^\text{act}$ reflects the reliability of a wireless channel. 
	
	After receiving the aggregated message, the server updates the global model and further broadcasts it in the downlink, which is identical to all devices. As transmit power and bandwidth are usually large for broadcasting, we consider it as the high SNR condition and neglect the distortion during broadcasting in the downlink.

	\begin{rem}[Outage Effect]
		It is possible that some devices disconnect from the server occasionally in practice, which is called as the \emph{outage effect}.  We consider disconnection as a special case of the channel-truncation, where all sub-channels are truncated. Moreover, when a device in outage reconnects to the server, it first receives the latest subspace broadcast from the server, and then continues to compute the local gradient and joins  the AirPCA again.
 	\end{rem}

	\subsection{Distributed PCA Problem and Algorithm}
	\subsubsection{Distributed PCA Problem} We assume a global dataset comprising $L$  samples is uniformly distributed over the $K$ devices. Let $\mathcal{D}_k$ denote local dataset of device $k$ generated by uniformly sampling the global dataset. The local datasets have a uniform size: $|\mathcal{D}_k| = \ell_0$ where $L = K \ell_0$. In this work we assume that the local datasets are acquired in advance and do not vary within the processing duration, which is a common setting adopted in \cite{Tut1,Iwen}. The distributed PCA problem is to find a low-dimensional subspace of the data space, called \emph{principal components},  to compress the distributed dataset under the criterion of minimum distortion. Let $d$ and $D$ with $D \gg d$ denote the dimensions of the principal components and data space, respectively. Let  the  $i$-th sample  be denoted as $\mathbf{x}_i\in \mathbb{R}^{D \times 1}$.  Moreover, $d$-dimensional principal components are represented by the unitary/orthogonal real matrix $\mathbf{W} \in \mathbb{R}^{D \times d}$. The sample   $\mathbf{x}_i$ can be approximated  using its projection onto the subspace, ${\mathbf{W}^T\mathbf{x}_i}$,  as $\mathbf{W}{\mathbf{W}^T\mathbf{x}_i}$. To minimize the approximation error, the distributed PCA problem can be formulated as: 
	\begin{align*}
	(\mathbf{P1}) \qquad \qquad \min_{\mathbf{W}}& \quad \frac{1}{L}\sum_{k=1}^K \sum_{i\in \mathcal{D}_k} \left\| \mathbf{x}_i - \mathbf{W}{\mathbf{W}^T\mathbf{x}_i} \right\|_2^2, \\
	\text{s.t.} & \quad \mathbf{W}^T\mathbf{W} = \mathbf{I},
	\end{align*}
	where $\mathbf{x}_i\!\in\! \mathbb{R}^{D \times 1}$ is a data sample, with $\mathbf{X} \!\in\! \mathbb{R}^{D \times L}$ as the aggregation.
	If all devices can upload their local data to the server,  Problem (${\mathbf{P1}}$) can be solved by applying SVD on the centralized dataset $\mathbf{X} \!=\! [\mathbf{x}_1, \mathbf{x}_2, \cdots, \mathbf{x}_L]$. However, for the distributed PCA scenario, direct data uploading is infeasible under  the data-privacy constraint. A different SGD-based solution is described as follows. 
	
	\begin{figure*}[t]\vspace{-0.5cm}
		\centering
		\includegraphics[height=110pt]{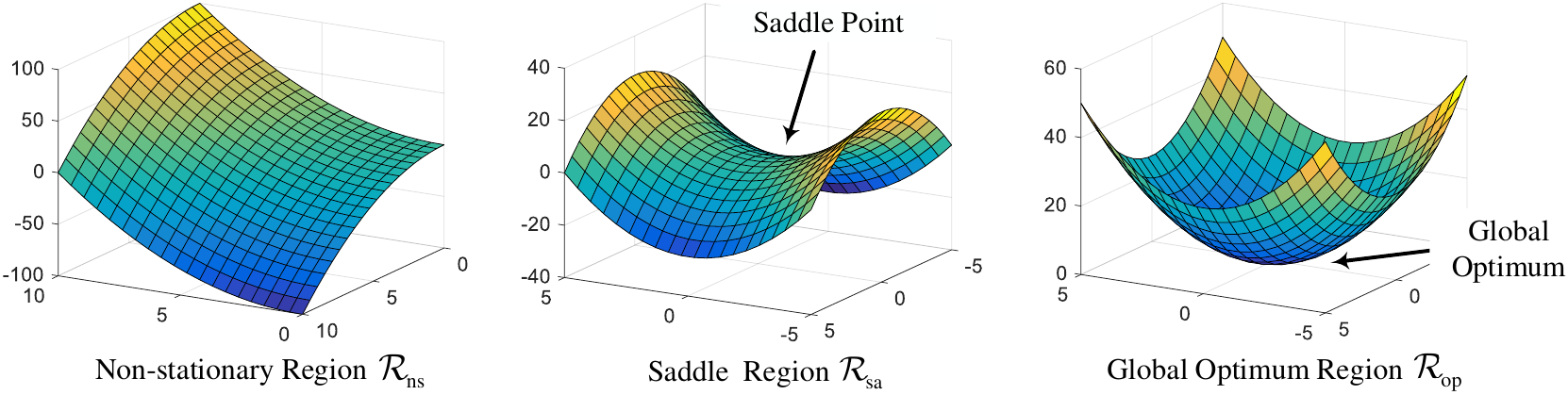} \caption{Three types of regions in a strict saddle function.}\label{ThreeRegions}\vspace{-0.3cm}
	\end{figure*}

	\subsubsection{Distributed PCA Algorithm}\label{Section:SGD} For ease of notation, let the objective function of Problem (${\mathbf{P1}}$) be denoted as 
	\begin{equation}
	F(\mathbf{W}) = \frac{1}{L}\sum_{i=1}^{L} \left\| \mathbf{x}_i - \mathbf{W}{\mathbf{W}^T\mathbf{x}_i} \right\|_2^2. 
	\end{equation}
	As  proved in \cite{SubspaceTracking},  $F(\mathbf{W})$ has stationary points in the form of $\mathbf{W} = \mathbf{U}_{d} \mathbf{Q}$, where the column vectors of $\mathbf{U}_{d} \in \mathbb{R}^{D \times d}$ are    $d$ distinct eigenvectors of the covariance matrix $\mathbf{R}=\mathbf{X}\mathbf{X}^T$ and $\mathbf{Q} \in \mathbb{R}^{d \times d}$ is an arbitrary unitary matrix.  If the Hessian matrix $\mathcal{H}(\mathbf{W}) =\nabla^2 F(\mathbf{W})$ has both positive and negative eigenvalues, then $\mathbf{W}$ is called a \emph{saddle point}. As  further proved  in \cite{SubspaceTracking}, all stationary points of $F(\mathbf{W})$ are saddle points, except for one where  $\mathbf{U}_{d}$ contains the $d$ dominant eigenvectors of $\mathbf{R}$. This point yields  the global minimum of $F(\mathbf{W})$. The above properties suggest that $F(\mathbf{W})$ comprises three types of region as illustrated in Fig. \ref{ThreeRegions}. Then the gradient-descent algorithm can be effective in solving the following optimization problem, which  is a simplified version of (${\mathbf{P1}}$) without its unitary/orthogonal constraint: 
	\begin{align}\label{groundtruth}
	(\mathbf{P2}) \qquad  \min_{\mathbf{W}}& \quad F(\mathbf{W}), \nonumber
	\end{align}
	if  the descent process can avoid being trapped at saddle points. A standard approach of escaping from a saddle point is to add artificial noise into the gradients \cite{Ge15}.  Then the column space of the optimal point, $\mathbf{W}^\star$, solves Problem (${\mathbf{P1}}$).

	As a special case of FL, the iterative algorithm of distributed PCA is based on SGD \cite{SGD}. To describe the algorithm, consider an arbitrary communication round of the algorithm. At its beginning, the server broadcasts the current principal components $\mathbf{W}$ to all devices for computing gradients based on all local data samples. To this end, the local objective function of device $k$ is given as $F_k(\mathbf{W}) = \frac{1}{\ell_0} \sum\limits_{i\in\mathcal{D}_k} \left\| \mathbf{x}_i - \mathbf{W}\mathbf{W}^T\mathbf{x}_i \right\|_2^2$. Moreover, define the data covariance matrix at device $k$ as $\mathbf{R}_k = \mathbf{X}_k\mathbf{X}_k^T$, where  the $D\times \ell_0$ matrix $\mathbf{X}_k$ comprises samples in the local dataset $\mathcal{D}_k$. 
	Then the local  gradient,  $F_k(\mathbf{W})$, is computed at device $k$ as 
	\begin{equation}\label{localGra}
	\nabla F_k(\mathbf{W})  = \frac{2}{\ell_0}\left[-2{\mathbf{R}_k}+{\mathbf{R}_k}\mathbf{W}\mathbf{W}^T+\mathbf{W}\mathbf{W}^T{\mathbf{R}_k} \right]\mathbf{W}.
	\end{equation}
	The devices  upload their local gradients to the server for aggregation and then updating the principal components $\mathbf{W}$. Note that the gradient of the global objective function $F(\mathbf{W})$ can be written in terms of local gradients as 
	\begin{align}\label{pureGra}
	\nabla F(\mathbf{W}) = \frac{1}{K}\sum_{k=1}^K \nabla F_k(\mathbf{W}).
	\end{align}
    However, the  received gradient is purposely perturbed by noise to escape from saddle points \cite{Ge15}: 
	\begin{equation}\label{noisyGlobal}
	\nabla \widehat{F}(\mathbf{W}) = \frac{1}{K}\sum\limits_{k=1}^K \nabla F_k(\mathbf{W}) + \mathbf{z},
	\end{equation}
	where $\mathbf{z}$ is a random vector  representing  noise. Then the principal components in the current round (say round $n$),  $\mathbf{W}_{n}$, are updated by the server: 
	\begin{equation}\label{Eq:Update}
	\mathbf{W}_{n+1} = \mathbf{W}_{n} - {\mu}\nabla \widehat{F}(\mathbf{W}_n),
	\end{equation}
	where ${\mu}$ is a fixed step-size. The above per-round procedure is repeated until $\mathbf{W}$ converges. 
			
	\subsection{AirPCA Implementation}
AirPCA implements  distributed PCA in an over-the-air aggregation system. The implementation of the $n$-th round is described as follows. To facilitate transmission over both in-phase and quadrature channels, the local and global gradients (matrices),  $\nabla F_k(\mathbf{W})$  and $\nabla F(\mathbf{W})$, are complex vectorized with mapping functions $g_k(\cdot)$ and $g(\cdot)$, where the resultants are denoted as $g_k(\mathbf{W})=\mathsf{vec}[\nabla F_k(\mathbf{W})]$ and $g(\mathbf{W})=\mathsf{vec}[\nabla F(\mathbf{W})]$, each comprising $c=\frac{D\times d}{2}$ elements. Given i.i.d. data distribution over devices, the following assumption of unbiased estimation is common in the literature of distributed learning and estimation (see e.g., \cite{Bertsekas,SGD}). 

\begin{assumption}[Unbiased Estimation] \emph{The local gradient computed at each device is an unbiased estimate of the global gradient: 
	\begin{equation}
	g_k(\mathbf{W}) = g(\mathbf{W}) + \boldsymbol{\Delta}_k,\qquad 1 \le k \le K, 
	\end{equation}
where the estimation error vector, $\boldsymbol{\Delta}_k$, is called \emph{data noise} and satisfies 	
	\begin{equation}\label{Eq:DataNoise}
\mathsf{E}[\boldsymbol{\Delta}_k] = \frac{1}{K}\sum_{k=1}^K\boldsymbol{\Delta}_k=\mathbf{0}, \quad \mathsf{E}[\boldsymbol{\Delta}_k\boldsymbol{\Delta}_k^H] \le \kappa^2 \mathbf{I},
	\end{equation}
for a given constant 	$\kappa^2$. 
}	
\end{assumption}	
Note from \eqref{Eq:DataNoise} that the data  noise $\{\mathbf{\Delta}_k\}$ at different devices are correlated.

To realize over-the-air aggregation, each device transmits its local gradient using linear analog modulation. Following the model in \cite{BB,ZY} for  i.i.d. data distribution, the symbols at device $k$, namely the  elements of the local  gradient $g_k(\mathbf{W})$, can be modeled as identically distributed random variables with mean $\eta$ and variance $\nu^2$;  the statistics are identical for all devices and are  known by them. To facilitate  power control in \eqref{powerConstraint}, each symbol that is not truncated   is normalized to have  zero mean and unit variance, i.e., ${\bf s}_{k,n} = \frac{g_k(\mathbf{W}_n)-\eta}{\nu}$ in the $n$-th round, and then transmitted over a sub-channel; otherwise, a symbol $0$ is transmitted. Being synchronized in time (using i.e., timing advance in 3GPP) and using truncated channel inversion in \eqref{truncation}, all devices simultaneously transmit their OFDM symbols with aligned boundaries to perform over-the-air aggregation. This yields the symbol vector as received by the server as 
	\begin{align}\label{ReceiveSig}
	\mathbf{y}_n = \sum_{k=1}^K \sqrt{{P^{\text{rx}}_n}} {\bf s}_{k,n}+\mathbf{z}_n.
	\end{align}
Then the received symbols are  de-normalized to give the elements of the noisy global gradient, denoted as $\widehat{g}(\mathbf{W}_n)$,  as 
	\begin{align}\label{NG}
	[\widehat{g}(\mathbf{W}_n)]_i &=  \frac{1}{K_n^{(i)}}\frac{\nu}{\sqrt{P^{\text{rx}}_n}} [\mathbf{y}_n]_i +\eta,  
	\end{align}
where $K_n^{(i)}$ is defined to be the number of devices transmitting  the $i$-th gradient element in the $n$-th round, with $\mathcal{K}_n^{(i)}$ denoting the set of devices, that is, $|\mathcal{K}_n^{(i)}| = K_n^{(i)}$. The number  follows a binomial distribution,  $K_n^{(i)} \sim B(K,\zeta^\text{act})$,  with $\zeta^\text{act}$ being the activation probability in \eqref{nonTruncation}. Equation  \eqref{NG} implies that $K_n^{(i)} $ is nonzero. This is reasonable since  $\Pr(K_n^{(i)} = 0) = (1-\zeta^\text{act})^{K}$, which is close to zero  when $\zeta^\text{act}$ is close to one and/or $K$ is large. The substitution of the normalization equation and \eqref{ReceiveSig} into \eqref{NG} gives the noisy global gradient as received by the server as 
	\begin{align}\label{FianlGradient}
	\widehat{g}(\mathbf{W}_n) = {g}(\mathbf{W}_n)+ \boldsymbol{\xi}_n,
	\end{align}
where the noise vector $\boldsymbol{\xi}_n$ combines channel and data noise and is defined element-wise as 
	\begin{align}\label{noise}
	[\boldsymbol{\xi}_n]_i =  \frac{1}{{K}_n^{(i)}}\bigg( \frac{\nu}{\sqrt{P^{\text{rx}}_n}}[\mathbf{z}_n]_i + \sum_{k\in\mathcal{K}_n^{(i)}}[\boldsymbol{\Delta}_{k}]_i\bigg), \quad 1\le i\le M.
	\end{align} 
	By de-vectorizing  $\widehat{g}(\mathbf{W}_n)$ in \eqref{FianlGradient} into the matrix $\widehat{F}(\mathbf{W}_n)$, the principal components are updated as in \eqref{Eq:Update}, completing the $n$-th round of AirPCA. 
	
	\section{Convergence Analysis for AirPCA}\label{convergenceAna}
In this section, the convergence of AirPCA is quantified in terms of descent speeds in different types of regions (see Fig. \ref{ThreeRegions}) and convergence accuracy. The results are useful for designing power control in the next section. 

	\subsection{Definitions and Assumptions}
For tractable analysis, several definitions and assumptions are given as follows. First, as discussed, the  objective function $F(\mathbf{W})$ of the PCA problem in ($\mathbf{P1}$) contains discrete saddle points, one global optimum without  local optimums. Such a function belongs to the family of  strict saddle functions defined as follows \cite{Ge15,StrictSaddle2}. 	
	\begin{defn}[Strict Saddle Function]\label{def1}
		A twice-differentiable function $F(\mathbf{W})$ is called $(\alpha,\gamma,\epsilon,\delta)$-strict saddle if for any point $\mathbf{W}$, at least one of the following is true
		\begin{enumerate}[1.]
			\item $\|\nabla F(\mathbf{W})\|\ge \epsilon$.
			\item Consider the Hessian matrix $\mathcal{H}(\mathbf{W}) =\nabla^2 F(\mathbf{W})$. Its minimum eigenvalue $\lambda_{\min}(\mathcal{H}(\mathbf{W})) \le -\gamma$ for some positive constant  $\gamma$. 
			\item Let $\mathbf{W}^\star$ be the point of global minimum of  $F(\mathbf{W})$ and $\delta$ and $\alpha$ given  positive constants.  In the $\delta$-neighbourhood $\{ \mathbf{W}\in \mathbb{R}^{D\times d}: \|\mathbf{W}-\mathbf{W}^\star\| \le \delta\}$, the function $F(\mathbf{W})$ is $\alpha$-strongly convex, i.e.,  $\lambda_{\min}(\mathcal{H}(\mathbf{W})) \ge \alpha$.			
		\end{enumerate}
	\end{defn}
The above definition allows the three types of regions of  $F(\mathbf{W})$ as illustrated in Fig. \ref{ThreeRegions} to be defined mathematically as follows. 

\begin{defn}[Region Types] A region of  $F(\mathbf{W})$ belongs to one of the following three types. 
\begin{itemize}
\item A  \emph{non-stationary region} [see Fig. \ref{ThreeRegions}(a)], denoted as $\mathcal{R}_\text{ns}$,  is one where condition 1) holds and thus can be defined as $\mathcal{R}_\text{ns} = \{\mathbf{W}\in\mathbb{R}^{D\times d}: \|\nabla F(\mathbf{W})\|\ge \epsilon\}$.
\item A  \emph{saddle region} [see Fig. \ref{ThreeRegions}(b)], denoted as $\mathcal{R}_\text{sa}$, is one where both conditions 1) and  2) hold and thus can be defined as  $\mathcal{R}_\text{sa} = \{\mathbf{W}\in\mathbb{R}^{D\times d}: \|\nabla F(\mathbf{W})\| < \epsilon; \lambda_{\min}(\mathcal{H}(\mathbf{W})) \le -\gamma \}$. 
\item A  \emph{global optimum region} [see Fig. \ref{ThreeRegions}(c)], denoted as $\mathcal{R}_\text{op}$, is one where condition 3) holds and thus can be defined as $\mathcal{R}_\text{op} = \{ \mathbf{W}\in \mathbb{R}^{D\times d}: \|\mathbf{W}-\mathbf{W}^\star\| \le \delta; \lambda_{\min}\mathcal{H}(\mathbf{W}) \ge \alpha\}$.
\end{itemize}
\end{defn}

For tractability, we make several typical assumptions on  $F(\mathbf{W})$ that introduce additional properties that usually  hold in practice (see e.g.,  \cite{Ge15}). 
\begin{assumption}\emph{The function $F(\mathbf{W})$ has several additional properties: 
	\begin{enumerate}
		\item (Boundedness) Both the function $F(\mathbf{W})$ and its gradient norm are  bounded: $\|F(\mathbf{W})\|\le B$ and   $\|g(\mathbf{W})\|\le C$ for all $\mathbf{W}$ and some constants $B$ and $C$. 
		\item (Smoothness) The function $F(\mathbf{W})$ is $\beta$-Lipschitz smooth:
		\begin{align}
		\|g(\mathbf{W}_1)-g(\mathbf{W}_2)\| \le \beta \|\mathbf{W}_1 - \mathbf{W}_2\|
		\end{align}
		for some positive constant $\beta$. 
		\item (Hessian smoothness) The Hessian of $F(\mathbf{W})$, $\mathcal{H}(\mathbf{W}) = \nabla^2 F(\mathbf{W})$,  is $\chi$-Lipschitz smooth:
		\begin{align}\label{Hessian}
		\|\mathcal{H}(\mathbf{W}_1)-\mathcal{H}(\mathbf{W}_2)\| \le \chi \|\mathbf{W}_1 - \mathbf{W}_2\|, 
		\end{align}	
		for some positive constant $\chi$. 
	\end{enumerate}
}
\end{assumption}
	
	\subsection{Characterizing Gradient Descent in Different Regions}\label{ConvergenceAnalysis}

	\subsubsection{Descent in non-stationary regions} The descent speed is measured by the expected reduction on the error function, termed the expected error reduction, over a given number of rounds. The descent speed in a non-stationary region is related to the receive signal power as well as other parameters as follows.

	\begin{theo}[Descent Speed in a Non-Stationary Region]\label{Compare} 
		Consider $n$-round gradient descent in  a non-stationary region, $\mathcal{R}_\text{ns}$, with the corresponding principal-component states $\left.\{\mathbf{W}_0, \dots, \mathbf{W}_{n-1} \}\subset \mathcal{R}_\text{ns}\right.$ and receive power controlled to be $\{P^\text{rx}_0, \dots, P^\text{rx}_{n-1}\}$. If the  step-size $\mu \le \frac{1}{\beta}$ with $\beta$ specifying the error-function smoothness, the expected error reduction  over the $n$ rounds can be lower bounded as 
		\begin{align}\label{CompareRegion1}
		 \!\!\!\!\mathsf{E}\left[F(\mathbf{W}_{0})\!-\!F(\mathbf{W}_{n})\right]
		\!\ge\!  n{\mu}   \!\left[\!\frac{\epsilon^2}{2} \!-\!  \frac{ \beta c \mu\kappa^2}{K \zeta^\text{act}} \!-\! \frac{3\beta c \mu\nu^2\sigma^2}{(K{\zeta^\text{act}})^2  \bar{P}^\text{rx}} \!\right], 
		\end{align}
where $\bar{P}^\text{rx} = \left[ \frac{1}{n}\sum_{m=0}^{n-1}\frac{1}{P^\text{rx}_m}\right] ^{-1}$. 
	\end{theo}
	\begin{proof}
		See Appendix A.
	\end{proof}

First of all, one can observe from \eqref{CompareRegion1} that the expected error reduction is proportional to $n\mu$, the order of descent distance. Next, the three terms enclosed by  the brackets at the right-hand side of \eqref{CompareRegion1} quantify the effects of the slopes of the error function, data noise, and channel noise respectively, which are explained as follows. The first term is proportional to the square of the minimum slope,  $\epsilon^2$, of the error function in $\mathcal{R}_\text{ns}$. Being negative, the second term reduces the descent speed by an amount proportional to the data-noise variance, $\kappa^2$, and inversely proportional to the expected number of devices performing over-the-air aggregation, namely $K\zeta^\text{act}$. As is well known in the literature of distributed learning, the latter scaling law results from more accurate distributed estimation due to a larger global dataset where there are more devices (see e.g., \cite{1Bit,Mingzhe}). 

The last term on the channel-noise effect is new in the literature of distributed PCA. One can observe that the descent-speed reduction due to channel noise is inversely proportional to $\frac{\bar{P}^\text{rx}}{\sigma^2}$, which can be interpreted as the expected receive SNR per device. This is obvious in the case of  fixed receive power, $\bar{P}^\text{rx}_m = P^\text{rx}_0$, for all $m$ for which $\frac{\bar{P}^\text{rx}}{\sigma^2}$ reduces to $\frac{P^\text{rx}_0}{\sigma^2}$. On the other hand, over-the-air aggregation results in the expected magnitude of the aggregated signal at the server increasing linearly with respect to the expected number of devices, $K\zeta^\text{act}$. Consequently, the expected SNR after aggregation is scaled up by $(K\zeta^\text{act})^2$, causing the channel-noise term in \eqref{CompareRegion1} to decrease as an inverse function of the factor.  In addition, as a sanity check, setting the channel noise variance $\sigma^2 = 0$ and the activation probability $\zeta^\text{act} = 1$, the result in Theorem \ref{Compare} converges to the existing one assuming reliable channels \cite{Ge15}. This also applies to Theorems \ref{saddle} and \ref{SurelyConverge}.

Based on the result in Theorem \ref{Compare}, we can draw the conclusion  that it is desirable to suppress the effect of channel noise by increasing the effective receive signal power, namely  $\bar{P}^\text{rx}$. In particular, given a power sequence $\{{P}^\text{rx}_0, \dots, {P}^\text{rx}_n\}$, if another sequence $\{\acute{P}^\text{rx}_0, \dots, \acute{P}^\text{rx}_n\}$ is larger than $\{{P}^\text{rx}_m\}$ element-wise, then $\{\acute{P}^\text{rx}_m\}$ leads to  larger expected reduction on the error function over the $n$ rounds.

	\subsubsection{Descent in saddle regions}	The descent speed in a saddle  region is related to the receive signal power as well as other parameters as follows.

	\begin{theo}[Descent Speed in a Saddle Region]\label{saddle} Consider $n$-round gradient descent in  a saddle region, $\mathcal{R}_\text{sa}$, with the corresponding principal-component states $\{\mathbf{W}_0, \dots, \mathbf{W}_{n-1} \} \subset \mathcal{R}_\text{sa}$ and finite receive power $\{P^\text{rx}_0, \dots, P^\text{rx}_{n-1}\}\subset [P^{\text{rx}}_{\min}, P^{\text{rx}}_{\max}]$. Define two constants $\mathcal{V}_{\max} =  \frac{\kappa^2}{K\zeta^\text{act}}+\frac{3\nu^2  \sigma^2}{K^2{\zeta^\text{act}}^2 P^{\text{rx}}_{\min}}$ and $\mathcal{V}_{\min} =  \frac{\nu^2  \sigma^2}{K^2 P^{\text{rx}}_{\max}}$. If the step-size and number of rounds satisify
		\begin{align}\label{stepsizeL2}
		\mu \ll \frac{1}{c\mathcal{V}_{\max}}, \qquad 
		n  > \frac{1}{2\mu\gamma}\log \left(6c\frac{\mathcal{V}_{\max}}{\mathcal{V}_{\min}}+1\right) = N_{\max}, 
		\end{align}
the expected error reduction over the $n$ rounds can be lower bounded as 	
		\begin{align}\label{L2_approx}
	    &\mathsf{E}[F(\mathbf{W}_{0})- F(\mathbf{W}_{n})] \nonumber\\
		\ge &\frac{\mu}{4}  \!\left[\!\frac{\kappa^2}{K\zeta^\text{act}} \!+\!  \mu\gamma\!\!\!\!\sum\limits_{m=0}^{n\!-\!N_{\max}\!-\!1}\!\!\frac{(1\!+\!\mu\gamma)^{2(n\!-\!m\!-\!1)}\nu^2 \sigma^2}{K^2 P^\text{rx}_m} \!+\!\frac{3\nu^2  \sigma^2}{K^2{\zeta^\text{act}}^2 \!P^{\text{rx}}_{\min}}\right].	
		\end{align}		
	\end{theo}
	\begin{proof}
		See Appendix C.
	\end{proof}
	
In a saddle region [see Fig. \ref{ThreeRegions}(b)], the gradient descent may be infeasible in some dimensions (e.g., one in which the error function is convex and the current point is the minimum); descent is guaranteed only in the dimension corresponding to the minimum eigenvalue $\lambda_{\min}(\mathcal{H}(\mathbf{W}))\le -\gamma$ which is concave. The result in Theorem \ref{saddle} shows that the gradient perturbation by the data-and-channel noise has the beneficial effect of warranting the expected descent (or equivalently strictly positive expected error reduction) if the step-size is sufficiently small and the number of rounds is sufficiently large. This results in a high probability of descending in the dimension corresponding to $\lambda_{\min}(\mathcal{H}(\mathbf{W}))$ due to the noise induced randomization of the descending  direction. In the brackets at the right-hand side of \eqref{L2_approx}, the first term and the last two terms represent the positive effects of data and channel noise on the descent speed, respectively, as opposed to their negative effects in a non-stationary region (see Theorem \ref{Compare}). 

An observation important for power control that can be made from  \eqref{L2_approx} is that enhancing the channel noise by reducing the receive signal power, $\{P^\text{rx}_m\}$, enhances the expected error reduction. Thus, it is desirable to set the power to its minimum, $P^\text{rx}_m = P^\text{rx}_{\min}$. As a result, the bound on the expected error reduction can be simplified as 
\begin{align}
&\mathsf{E}[F(\mathbf{W}_{0})- F(\mathbf{W}_{n})]\nonumber\\ 
		\ge &\frac{\mu}{4}  \left[\frac{\kappa^2}{K\zeta^\text{act}} +  \left(\phi(\mu,n)+\frac{3}{(\zeta^\text{act})^2}\right) \frac{\nu^2 \sigma^2}{K^2 P^{\text{rx}}_{\min}} \right].	
\end{align}	
where $\phi(\mu,n) = \frac{(1+\mu\gamma)^{2n}-(1+\mu\gamma)^{2N_{\max}}}{2+\mu\gamma}$. On the other hand, it should be emphasized that the receive signal power should not be too low as too strong noise can make the aggregated gradient (or equivalently the descent direction) completely random and thereby make it impossible to truly escape from a saddle point in the long term, namely repeatedly  returning to the point. 
	
	\subsubsection{Convergence likelihood and accuracy}	
	
The results in  Theorems \ref{Compare} and \ref{saddle} show that the gradient descent of  AirPCA is not trapped in any non-stationary or saddle region. Consequently, the descent path eventually enters the optimum region almost surely, leading to learning convergence. The likelihood of convergence can be mathematically characterized in the following theorem, where the constants $\mathcal{V}_{\max}$ and $N_{\max}$ follow those  defined in Theorem \ref{saddle}. 

\begin{theo}\label{SurelyConverge}
Consider  $N$-round gradient descent for AirPCA from an arbitrary initial point  and  a step-size $\mu$ satisfying $\mu \ll \frac{1} {c \mathcal{V}_{\max}}$ and $\mu < \frac{\epsilon^2} {4 \beta c \mathcal{V}_{\max}}$. Let $\mathcal{E}_N$ denote the event  that the descent path enters the optimum region within $N$ rounds: $\mathcal{E}_N = \left\{\text{There exists some $n$ such that $0\!\le\! n \!\le\! N\!-\!1$}\right.\\ \left.\text{and $\mathbf{W}_n\in \mathcal{R}_\text{op}$}.\right\}$. If $N \!=\! mN_{\max}$ with $m\!\in\!\mathds{N}^+$, the probability of $\mathcal{E}_N$  can be lower bounded as
	\begin{align}
	\Pr(\mathcal{E}_{N}) \ge 1-\frac{12B}{(m+1) \mu \rho \mathcal{V}_{\max}},
	\end{align}
	where the constant  $\rho = \min\left\{\frac{2\beta c}{\gamma} \log(6c\frac{\mathcal{V}_{\max}}{\mathcal{V}_{\min}}+1), 1 \right\}$, and $B$ is the upper-bound on the error-function norm.
	\end{theo}

	\begin{proof}
		See Appendix D.
	\end{proof}

Theorem \ref{SurelyConverge} shows that if the step-size $\mu$ is sufficient small and the number of rounds is sufficiently large, the convergence is guaranteed in probability by ensuring $\Pr(\mathcal{E}_{N})$ close to one. 
Although it is possible for the descent path to escape from the optimum region due to  accidental strong noise, it will return to  $\mathcal{R}_\text{op}$ almost surely  according to Theorem \ref{SurelyConverge}. 

A standard analytical method for SGD can be applied to characterize the convergence accuracy. For instance, by  similar analysis as in \cite{Theo3Ref,Ge15}, it can be shown that if the number of rounds is sufficiently large, the distance between the learned principal components, $\mathbf{W}_{n}$,  and the optimal point $\mathbf{W}^\star$, namely $\|\mathbf{W}_{n} - \mathbf{W}^\star \|^2$, is linearly proportional to $\mu q \sqrt{\mathsf{E}[\|\boldsymbol{\xi}\|^2]}$ where $\boldsymbol{\xi}$ is the data-plus-channel noise sample in \eqref{FianlGradient}.

	\section{Region-Adaptive Power Control}\label{PowerControl}
Building on the convergence analysis in the preceding section, the scheme of region-adaptive power control to accelerate AirPCA is designed in this section. The scheme comprises of two component schemes, online detection of descent regions and online power control. They are described sequentially in the following subsections. 

	\subsection{Online Detection of Descent Regions}

Online detection of the type of the current descent region is the key for realizing the proposed scheme of region-adaptive power control. The main challenge lies in detecting a saddle region due to the conflict. Consider an arbitrary round, say the $n$-th round. On one hand, it follows from  the region's definition that its type can be detected by estimating  the minimum eigenvalue of the Hessian matrix, namely  $\lambda_{\min}(\mathcal{H}(\mathbf{W}_n))$, and evaluating its value  against some  given negative constant  $-\gamma$. If a saddle region is detected, channel noise should be enhanced so that the descent path can escape from being trapped at the saddle point. On the other hand, the estimation of the Hessian matrix $\mathcal{H}(\mathbf{W}_n)$ is difficult. Specifically, at best the server has the knowledge of one descent path that provides only  partial knowledge of $\mathcal{H}(\mathbf{W}_n)$ but the full knowledge is required for computing its eigenvalues.  Due to the difficulty of detecting a saddle region based on its definition, we propose a simple and effective online detection scheme described as follows. Again, consider the $n$-th round where the norm of the aggregated gradient $\|\widehat{g}(\mathbf{W}_{n})\|$ is found to be below a given threshold $\epsilon$ while that in the preceding round is above $\epsilon$. This indicates the descent path is entering a region which is either a saddle or an optimum region. By default, the region is detected as a saddle region and then the receive signal power is reduced to amplify the noise effect for the path to escape from a saddle point.  Given a decreased SNR, the gradient descent is continued for $N_0$ rounds where $N_0$ is a design parameter. Then the resultant  expected error  reduction over $N_0$ rounds, namely $[F(\mathbf{W}_{n-N_0}) - F(\mathbf{W}_{n})]$,  is evaluated against a positive threshold $f_0$. If the detection of a saddle region is correct, the escape from the saddle point should lead to substantial error reduction according to Theorem \ref{saddle} and thus $[F(\mathbf{W}_{n-N_0}) - F(\mathbf{W}_{n})] \geq f_0$. Otherwise, the detection is incorrect and the region should be the optimum region. Assuming that the decreased SNR is not too low so that the descent path remains within the region after $N_0$ rounds, the power control is adapted to the optimum region to reduce  noise to ensure a  small error after convergence. Last, the detection of a non-stationary region is straightforward and the criterion is $\|\widehat{g}(\mathbf{W}_{n})\| \ge \epsilon$.  

The scheme of online descent-region detection is summarized in Algorithm \ref{alg:algorithm2}. 

\begin{algorithm}[!t]
\caption{Online Descent-Region Detection.\label{alg:algorithm2}}	
	\KwIn{Error reduction threshold $f_0$ and testing round number $N_0$.}
	\KwOut{Region detection $\Theta_{n}$.}  

	{\bf Initialize} $n=0$, and $\Theta_{n} = \mathcal{R}_\text{ns}$;

	\While{not converge}{
		
		Calculate $\|\widehat{g}(\mathbf{W}_n)\|$;
		
		\If{$\|\widehat{g}(\mathbf{W}_n)\| < \epsilon$ and $\Theta_n = \mathcal{R}_\text{op}$}{
			
		Detect $\Theta_{n+1} = \mathcal{R}_\text{op}$;	
			
		Reduce noise and continue one round;	
			
		Set $n=n+1$;		
	    }
    
    	\ElseIf{$\|\widehat{g}(\mathbf{W}_n)\| < \epsilon$ and $\Theta_n \ne \mathcal{R}_\text{op}$}{
    	
    	Detect $\Theta_{n+1}, \dots, \Theta_{n+N_0} = \mathcal{R}_\text{sa}$ by default;	
    			
    	Continue $N_0$ rounds;			
    	
    	Set $n=n+N_0$;
    	
    	Calculate $[F(\mathbf{W}_{n-N_0}) - F(\mathbf{W}_{n})]$ by aggregation;
    	
    	\If{$[F(\mathbf{W}_{n-N_0}) - F(\mathbf{W}_{n})] < f_0$}{
    		
    		Detect $\Theta_{n} = \mathcal{R}_\text{op}$;	}
    }
		\Else{ 			
		Detect $\Theta_n = \mathcal{R}_\text{ns}$;
					
		Continue one round;			
		
		Set $n=n+1$;
		}		
	}	
\end{algorithm}

	\subsection{Online Power Control}\label{ProScheme}
Building on the preceding scheme of online region detection, the principle of region-adaptive power control is to reduce receive signal power when  the descent path enters a saddle region but increase the power if the path enters a non-stationary or optimum region. The former helps the path escape from a saddle point using channel noise (see Theorem \ref{saddle}) while the latter overcomes the noise to approach the steepest descent (see Theorem \ref{Compare}). 

Consider the case where a saddle region, $\mathcal{R}_\text{sa}$,  is detected. Then truncated channel inversion in \eqref{truncation} is controlled by each device so that the receive signal power is fixed at a chosen parameter  $P^\text{rx}_{\min}$ throughout the sojourn in $\mathcal{R}_\text{sa}$. Mathematically,  $P^\text{rx}_{n} = P^\text{rx}_{\min}$ for all $\mathbf{W}_n \in \mathcal{R}_\text{sa}$. The  parameter  $P^\text{rx}_{\min}$ should be chosen carefully, e.g., using experiments in the sequel.  As discussed, though $P^\text{rx}_{\min}$ should be sufficiently low so as to exploit the noise effect, its being too low can jeopardise finding the right descent path. Under the average power constraint in \eqref{powerResult},  it is necessary to choose $P^\text{rx}_{\min}$ to be smaller than the maximum average receive power $\bar{P}^\text{rx}_{\max}$, which saves power for use in other types of regions. Let $N_{\text{sa}}$ denotes the number of rounds for descent within $\mathcal{R}_\text{sa}$. Then the power saving is given as $N_{\text{sa}}(\bar{P}^\text{rx}_{\max} - P^\text{rx}_{\min})$.

Next, consider where either a non-stationary or optimum region is detected, denoted as $\mathcal{R}_\text{ns/op}$. The power-control policy is identical for both types of regions. Its key feature is to spend the accumulated power saving on accelerating the  descent in the current region. Let $n_0, n_1, \cdots, n_{N-1}$ denote the  rounds within $\mathcal{R}_\text{ns/op}$ with $N$ representing the total number of rounds. The accumulated saving can be written as $P^\text{rx}_{\text{save}} = \sum_{m = 0}^{n_0-1}(\bar{P}^\text{rx}_{\max} - P^\text{rx}_{\min})$. We propose that the receive signal power in the current region is controlled as $P^{\text{rx}}_n =  \bar{P}^{\text{rx}}_\text{max} + a_n {P^{\text{rx}}_{\text{save}}}$ for $n_0 \leq n \leq n_{N-1}$. The coefficients $\{a_n\}\!\subset\! [0, 1]$ are called power-spending coefficients and set using one of the following two designs. 
	\begin{enumerate}
		\item \emph{One-shot power-saving spending:} All of the accumulated power saving is used  in the first round upon the descent path entering $\mathcal{R}_\text{ns/op}$, namely $a_{n_0} = 1$ and $a_{n} = 0$ for $n = n_1, \cdots, n_{N-1}$. In other words, $P^{\text{rx}}_{n_0} =  \bar{P}^{\text{rx}}_\text{max} + {P^{\text{rx}}_{\text{save}}}$ and $P^{\text{rx}}_{n} =  \bar{P}^{\text{rx}}_\text{max}$  for $n = n_1, \cdots, n_{N-1}$. 

		\item \emph{Gradual power-saving spending:} The accumulated power saving is spent  over all rounds following  $a_{n_j} = (1-q)q^{j}$ for $0\leq j \leq N-1$ with $q \in (0, 1)$. Since $\sum_{n = n_0}^{n_{N-1}} a_{n_j}\leq 1$, all of the accumulated power saving is spent in $\mathcal{R}_\text{ns/op}$ if $N$ is large or $q$ is close to zero. Otherwise only part of the saving is used and the remaining is kept for subsequent regions along the descent path. 
		
	\end{enumerate}

Last, it should be emphasized that the above scheme for online power control guarantees that the average power constraint is satisfied. Moreover, the computation complexity of the power control scheme is $\mathcal{O}(Dd)$ for each round.

\section{Experimental Results}\label{Sim}

\subsection{Experiment Settings}
The default settings are as follows unless specified otherwise.  Three popular real training  datasets, MNIST, CIFAR-10 and AR, are used in separate experiments. After vectorization, the dimensions,  $D$,  of a single  data sample are $784$, $3072$ and $4800$, respectively. The reduced data dimensions are  set as $d=10$. For each experiment, $500$ data samples are randomly drawn from the dataset and  uniformly distributed over  $50$ devices. The parameters of  truncated channel inversion in \eqref{truncation} are set as  $G=0.2$ and the resultant  activation probability is  $\zeta^\text{act}=\exp(-0.2)$. The step-size is  $\mu=0.005$ for MNIST dataset and $\mu=0.02$ for CIFAR-10 and AR datasets, which are optimized by trials. Given the learned principal components, the PCA error is  evaluated using a testing dataset comprising $500$ samples randomly drawn from the  used dataset. The number of sub-channels is $M=1000$ with interval of $15$ kHz. The channel coefficients are identically distributed complex Gaussian variables with zero mean and unit variance for each sub-channel. We set the average transmit power for all devices as $\bar{P} = 26$ dBm and the noise power as $-100$ dBm over the whole bandwidth. For SGD, the principal components are initialized using the $d$-dimension identity matrix $\mathbf{I}_d$:  $\mathbf{W}_0 = \left[ \mathbf{I}_d,\mathbf{0}\right]^T $. 

Two benchmarking schemes are considered. One is fixed receive power: $P^{\text{rx}}_{n} = \bar{P}^{\text{rx}}_{\max}$ for all $n$. The other is  the ideal case of centralized PCA using  SVD.

	\subsection{Region-Adaptive Power Control}
	
	
	\begin{figure}[t]
			\centering
			\subfigure[Comparison with noise-free AirPCA.]{\includegraphics[height=5.6cm]{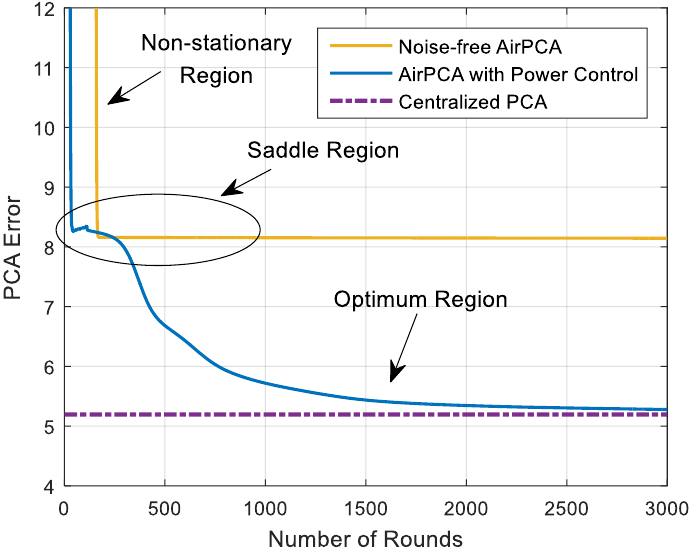}}
			\subfigure[Comparison with AirPCA using fixed power.]{\includegraphics[height=5.6cm]{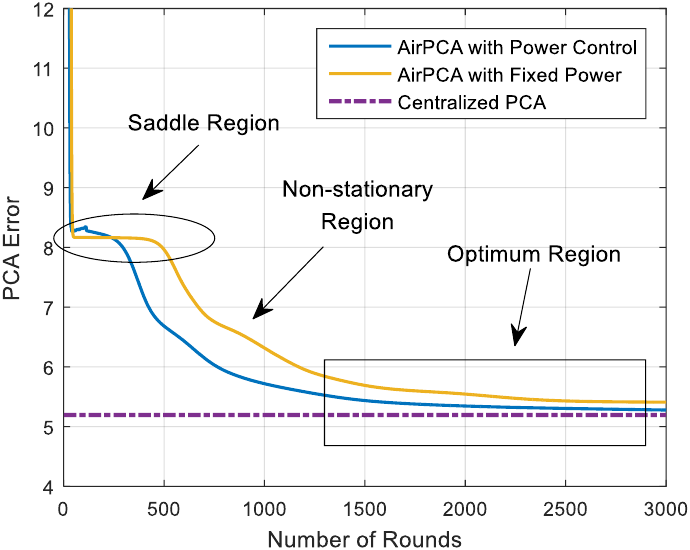}}
		\caption{The usefulness  of channel noise for AirPCA to escape from saddle points can be observed by comparing AirPCA with region-adaptive power control, noise-free AirPCA, and AirPCA with fixed power. The MNIST dataset is used.}\label{MNIST:Noise}\vspace{-0.5cm}
	\end{figure}	

\begin{figure}[t]
	\centering
	\subfigure[CIFAR-10 Dataset.]{\includegraphics[height=5.6cm]{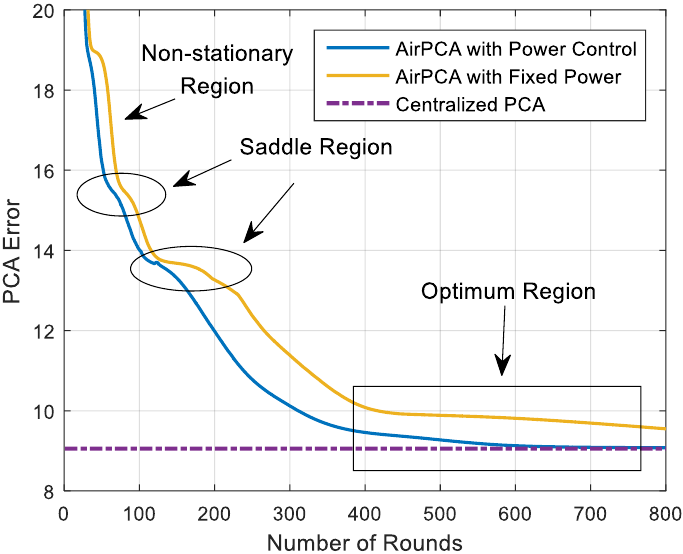}}
	\subfigure[AR Dataset.]{\includegraphics[height=5.6cm]{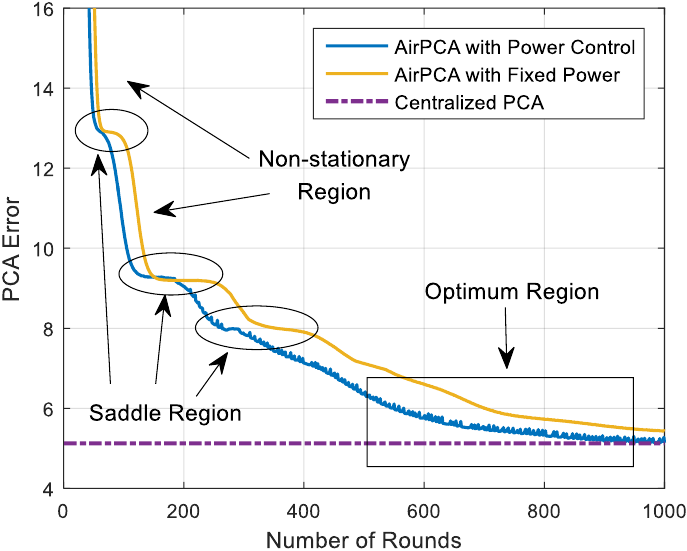}}
	\caption{Learning performance comparison using the CIFAR-10 and AR datasets with the step-size set as $\mu=0.02$.}\label{CIFAR}\vspace{-0.5cm}
\end{figure}	
To demonstrate the benefit of channel noise, the curves of PCA error versus  number of rounds  are plotted in Fig.~\ref{MNIST:Noise}(a) for  AirPCA with channel noise and region-adaptive power control (labeled as ``AirPCA with Power Control'')  and AirPCA without channel noise (labeled as ``Noise-free AirPCA''). The curve for centralized PCA is also plotted for comparison. The learned principal components of AirPCA with noise are observed to converge to those of centralized PCA after about $2000$ rounds while those in the noise-free case fail to do so. The reason is that the (gradient) descent path of the former escapes from the saddle point with the help of channel noise while that of the latter  is trapped at the point. Next, the learning performance of AirPCA with region-adaptive power control, AirPCA with fixed power, and centralized PCA are compared in Fig.~\ref{MNIST:Noise}(b), where the curves of PCA error versus number of rounds are plotted. One can observe that the proposed power-control scheme effectively accelerates the convergence w.r.t. the case with fixed power. For instance, to achieve the PCA error $7\%$ (i.e., error of $5.6$) above the level of centralized PCA (i.e., error of $5.2$), the learning latency is about {\bf $\mathbf{1170}$ rounds} compared with {\bf $\mathbf{1740}$ rounds} for AirPCA with fixed power, namely {\bf $\mathbf{33\%}$ reduction} in learning latency. Furthermore, the learning performance is also compared using two other datasets, CIFAR-10 and AR, in Fig.~\ref{CIFAR}. 
As in the last comparison, one can make the same observation that region-adaptive power control accelerates convergence. Last, it is worth mentioning that the initial part of the descent process for MNIST (see Fig. \ref{MNIST:Noise}) is relatively abrupt as compared with those for the other datasets (see Fig. \ref{CIFAR}). The  reason is that the  data samples in MNIST are black-and-white images of handwritten letters for which  the data information is more  concentrated in the subspace of principal components than that of CIFAR-10 and AR, composed of colorful and gray-scale images, respectively. In general, the descent speed depends on the power distribution of the components, which varies w.r.t. different datasets.

	
Next, in Fig.~\ref{PEC},  we compare the two designs of power-spending coefficients, namely one-shot and gradual power-saving spending,  in the proposed scheme of region-adaptive power control in terms of their effects on the learning performance. Both the MNIST and CIFAR-10 datasets are used and the descent step-sizes are set as $\mu=0.005$ and $\mu=0.02$, respectively. One can see that gradual spending of power-saving in the non-stationary and optimum regions with an optimized parameter (i.e., $q = 0.8$) achieves faster convergence than the one-shot schemes or gradual schemes with alternative values for $q$ (e.g., $0.5$ or $0.995$). It can be  observed  that their  different effects on the convergence lie in the stationary and optimum regions but not in the saddle regions where signal power is unaffected by the power-spending coefficients. Furthermore, the convergence accuracies are unaffected. 
	
	\begin{figure}[t]
		\centering
		\subfigure[MNIST Dataset.]{\includegraphics[height=5.6cm]{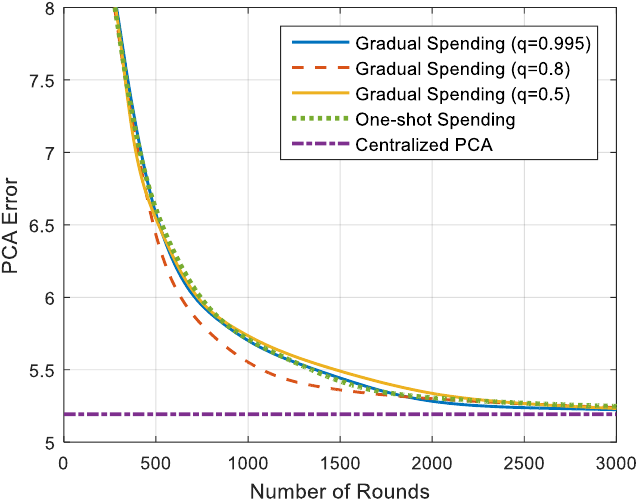}}
		\subfigure[CIFAR-10 Dataset.]{\includegraphics[height=5.6cm]{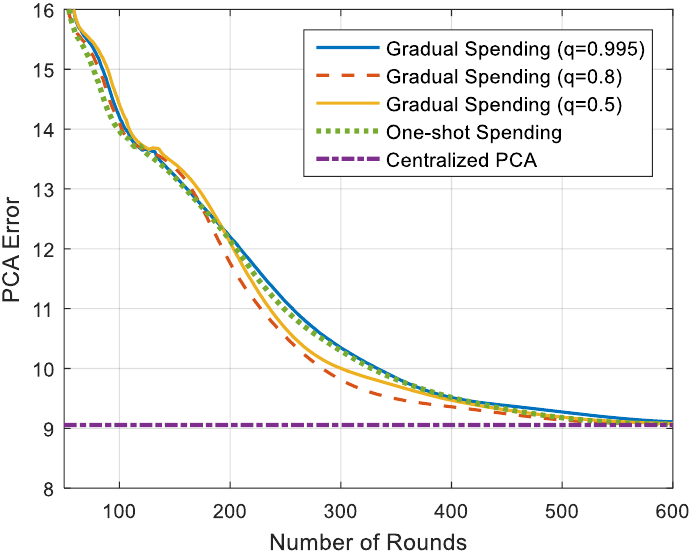}}
		\caption{The effects of power-spending coefficients on the learning performance of  AirPCA with region-adaptive power control. The datasets MNIST and  CIFAR-10 are used.}\label{PEC}\vspace{-0.3cm}
	\end{figure}

\subsection{Effects of Other System Parameters}


Considering AirPCA with region-adaptive power control, the curves of PCA error versus number of rounds are plotted in Fig. \ref{UserNumber}(a) for a varying number of devices, $K = \{10, 20, 50\}$. Each device is provided with $10$ data samples randomly drawn from the dataset. Thus, the total data used in AirPCA/centralized-PCA  are proportional to the number of devices. We take the CIFAR-10 dataset for experiment with step-size $0.02$. For $K = \{20, 50\}$, the learning performance is better for larger number of devices. On the other hand, when the number is small (e.g., $K = 10$), SGD-based AirPCA fails to converge due to the joint effect of limited data and insufficient aggregation gain that suppresses channel noise [see \eqref{noise}]. In contrast, centralized  PCA using SVD does not encounter such a problem.  One possible solution to prevent the divergence is to reduce the  step-size in AirPCA at the cost of slowing down the convergence. 

\begin{figure}[t]
			\centering
			\subfigure[Effect of different number of devices.]{\includegraphics[height=5.6cm]{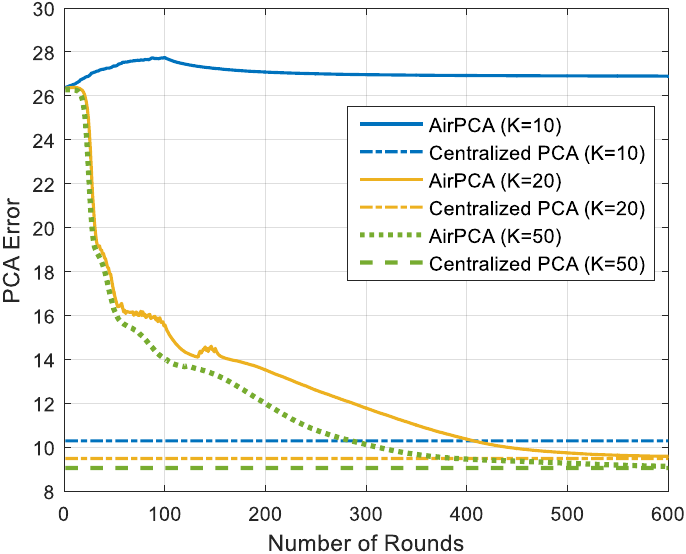}}
			\subfigure[Effect of different truncation thresholds.]{\includegraphics[height=5.6cm]{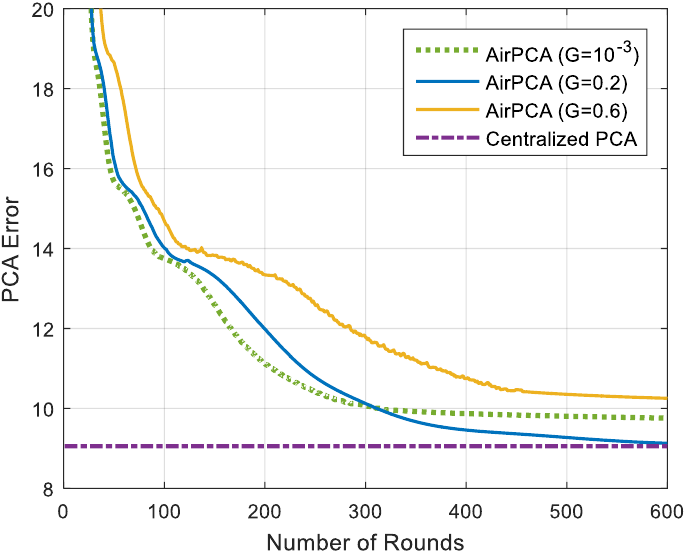}}
\caption{Effect of the number of devices and the truncation threshold on the learning performance of AirPCA with region-adaptive power control based on the CIFAR-10 dataset.}\label{UserNumber}\vspace{-0.3cm}
	\end{figure}

Next, we investigate the effect of channel-truncation threshold in \eqref{truncation}, $G$, on the learning performance of AirPCA with region-adaptive power control. To this end, 
the curves of PCA error versus number of rounds are plotted in Fig. \ref{UserNumber}(b) for a varying value of the truncation threshold $G = \{0.001, 0.2, 0.5\}$ for the CIFAR-10 dataset. Note that $G$ controls the expected ratio of truncated sub-channels. One can see that setting $G$ too small or too  large can result in  divergence. The former is due to too small receive signal power under the constraint of magnitude alignment across active sub-channels for over-the-air aggregation [see  \eqref{truncation}]; the latter is due to too many truncated sub-channels that severely distort the uploaded local gradients. This suggests the need of optimizing $G$,  for which finding a tractable approach is not obvious  but a topic warranting future work.

Define AirPCA latency as the required number of rounds to achieve the target PCA error relative to that of the ideal case of centralized PCA. To this end,  define the \emph{error ratio} $\varpi = \frac{\text{error for AirPCA}}{\text{error for centralized PCA}}-1$.  In Fig. \ref{TruncationChange}(a), we compare AirPCA latency for achieving different error ratios by varying the channel truncation threshold. It shows that the threshold being too large increases the latency, which is because deactivating  more  devices not only reduces the global  dataset used for AirPCA, but also results in weaker aggregation gain. The  results in  Fig. \ref{TruncationChange}(a) show the need of optimizing the threshold e.g., a truncation threshold in $(0, 0.2]$ is a preferred  choice. On the other hand, the effect of receive power used in the saddle regions on the AirPCA latency is also demonstrated  in Fig. \ref{TruncationChange}(b). It shows that higher  power slows  convergence, which is aligned with the finding in Theorem 2. Nevertheless, we can also see that too low receive power also leads to slow convergence. The reason is that strong noise perturbation  randomizes the gradient direction and can result in an undesired ascent direction.

\begin{figure}[t]
	\centering
	\subfigure[Effect of different truncation thresholds.]{\includegraphics[height=5.6cm]{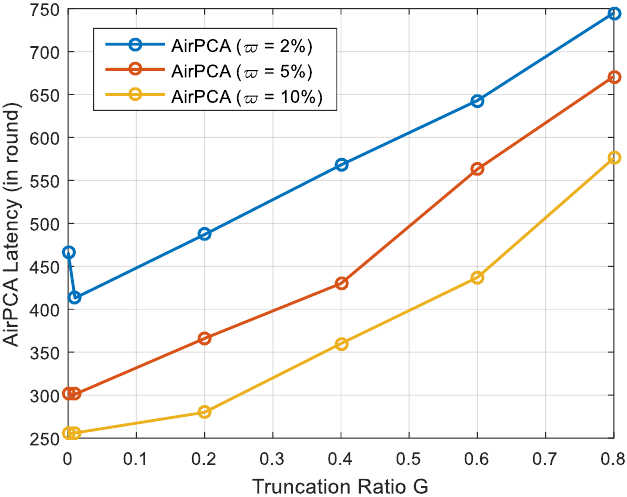}}
	\subfigure[Effect of different receive power for saddle regions.]{\includegraphics[height=5.6cm]{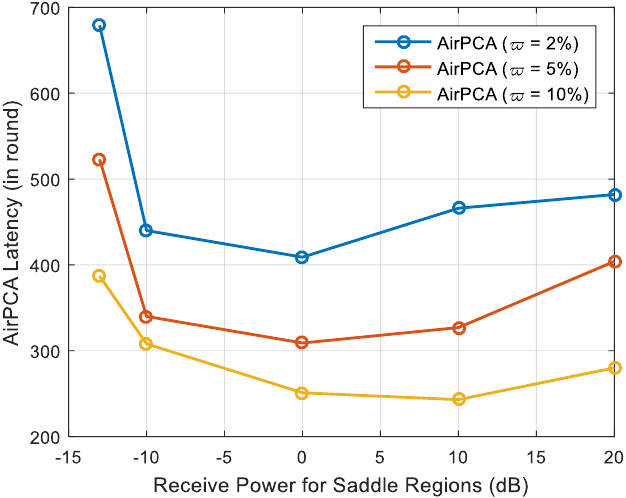}}
	\caption{AirPCA latency  for achieving different target PCA error ratios  with varying truncation thresholds $G$ and receive power for saddle regions. The CIFAR-10 dataset is considered with the step-size set as  $\mu=0.02$.}\label{TruncationChange}
\end{figure} 


\begin{figure}[t]
	\centering
	\subfigure[Comparison with the one-shot method in \cite{improvedPCA} on the total processing latency with $d=10$.]{\includegraphics[height=5.6cm]{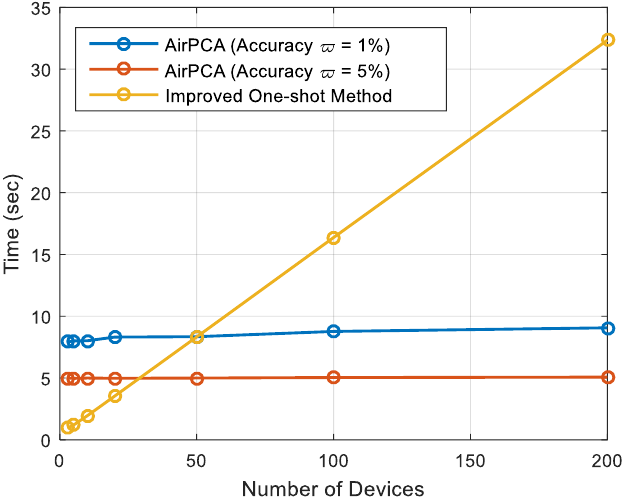}}
	\subfigure[Comparison with the power method in \cite{OTAPowerMethod} with SNR = $10$ dB.]{\includegraphics[height=5.6cm]{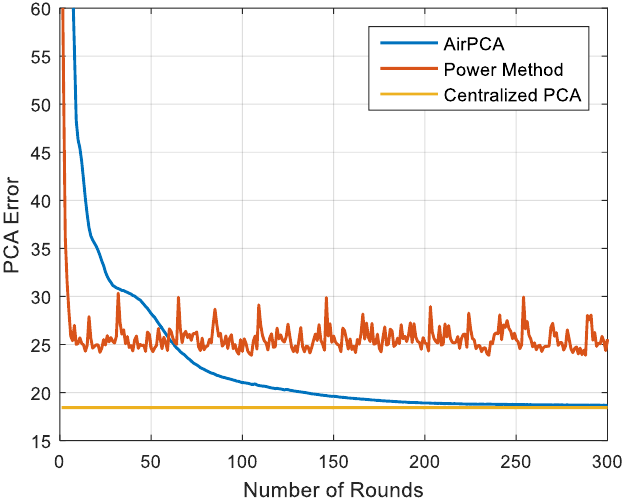}}
	\caption{The comparison between AirPCA and the state-of-the-art one-shot method in \cite{improvedPCA} and power method in \cite{OTAPowerMethod}. The CIFAR-10 dataset is used.}\label{CIFAR:Compare}\vspace{-0.3cm}
\end{figure}

Moreover, comparisons between the proposed AirPCA and the state-of-the-art one-shot method in \cite{improvedPCA} and power method in \cite{OTAPowerMethod} are also provided in Fig. \ref{CIFAR:Compare}. We assume each device acquires $30$ data samples in advance. The AirPCA and the power method feature negligible communication latency by applying over-the-air aggregation, while in the one-shot method we assume all the devices share a channel with constant transmission rate $8$ Mbits/s for local estimates uploading. Firstly, in Fig. \ref{CIFAR:Compare}(a) we show that the proposed AirPCA beats the one-shot method when the number of devices grows large, e.g., larger than $50$. The total processing latency of AirPCA remains to be $5\sim10$ seconds when the number of devices increases. The accuracy of the one-shot method is $\varpi < 1\%$ by adding $3$ redundant dimensions on the local subspace which helps suppress the biased error. On the other hand, in Fig. \ref{CIFAR:Compare}(b), we compare the proposed AirPCA with the power method on the convergence property, where SNR = $10$ dB and $30$ devices are involved. It clearly shows that the power method is sensitive to the noise perturbation while AirPCA guarantees the convergence to the global optimum.

Last but not least, in Fig. \ref{BatchSize}, we further show the effectiveness of the proposed AirPCA using the mini-batch approach at devices. In the current case, using  the CIFAR-10 dataset, AirPCA is performed  involving $20$ devices, each provisioned with  $30$ data samples. In each round, each device randomly selects  a mini-batch with a varying size  to compute  the local gradient. From Fig. \ref{BatchSize}, we can observe that though smaller mini-batch sizes result in slower convergence, they all lead to the same learning performance as the full-batch approach after convergence. The reason is that mini-batches generated by uniformly sampling the global dataset are representative of the latter's  distribution.

\begin{figure}
	\centering
	\includegraphics[height=5.6cm]{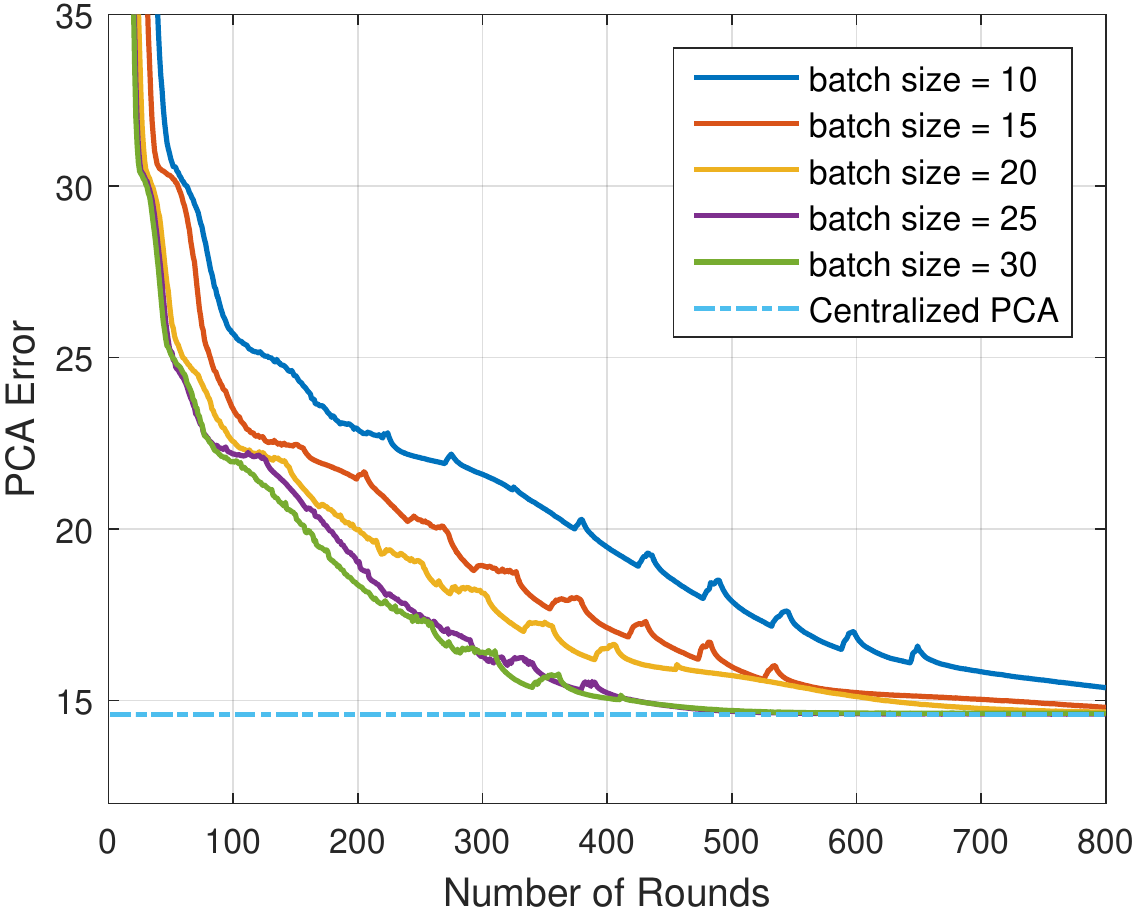}
	\caption{Effect of the batch size when applying AirPCA with SGD in each round.}\label{BatchSize}\vspace{-0.3cm}
\end{figure}

\section{Concluding Remarks}\label{Conclusion}
	
In this paper, we have proposed AirPCA that applies over-the-air FL to distributed PCA. Targeting the system, the key contribution of this paper is  the new idea of exploiting channel noise to accelerate convergence by escaping from saddle points. The idea has been materialized by designing an online power-control scheme featuring  descent-region awareness. While this work focuses on distributed PCA, the mentioned idea is general and useful for improving the performance of different types of Internet-of-Things and edge intelligence systems involving the operation of data aggregation such as distributed training of deep neural networks and distributed inference. Furthermore, the current AirPCA framework can be extended to advanced wireless techniques such as multi-antenna transmissions, millimeter-wave communications,  and interference limited systems. In particular, for the latter, the possibility of exploiting interference for accelerating convergence warrants investigation.

\section{Appendix}
\subsection{Proof of Theorem \ref{Compare}}\label{A}
	Given the real vectorized gradient $\vec{g}(\mathbf{W}_n)$ and following the gradient descent process in \eqref{Eq:Update} with the constraint on the step-size $\mu \le \frac{1}{\beta}$,
	\begin{align}\label{simpleExpression}
	&\mathsf{E}\left[F(\mathbf{W}_{n})-F(\mathbf{W}_{n+1})\right]\nonumber\\ \ge& -\vec{g}(\mathbf{W}_n)^T \mathsf{E}\left[\vec{\mathbf{w}}_{n+1}-\vec{\mathbf{w}}_n\right]-\frac{\beta}{2} \mathsf{E}\left[\|\vec{\mathbf{w}}_{n+1}-\vec{\mathbf{w}}_n\|^2\right],\nonumber\\
	=&\!-\!\vec{g}(\mathbf{W}_n)^T \mathsf{E}\!\left[\!-\mu(\vec{g}(\mathbf{W}_n)\!+\!\vec{\boldsymbol{\xi}}_n)\right]\!\!-\!\!\frac{\beta}{2} \mathsf{E}\!\left[\!\| \!\!-\!\mu(\vec{g}(\mathbf{W}_n)\!+\!\vec{\boldsymbol{\xi}}_n)\|^2\right],\nonumber\\
	=&(\mu-\frac{\beta\mu^2}{2})\|\vec{g}(\mathbf{W}_n)\|^2 - \frac{\mu^2\beta}{2}\mathsf{E}\left[\vec{\boldsymbol{\xi}}_n^T\vec{\boldsymbol{\xi}}_n\right],\nonumber\\
	\ge& \frac{1}{2}\mu\|\vec{g}(\mathbf{W}_n)\|^2 - \frac{\mu^2\beta}{2}\mathsf{E}\left[\vec{\boldsymbol{\xi}}_n^T\vec{\boldsymbol{\xi}}_n\right],
	\end{align}
	where $\vec{\mathbf{w}}$ is the real vectorization of matrix $\mathbf{W}$, and $\vec{\boldsymbol{\xi}}_n$ is the real vectorization of the data-plus-channel noise in \eqref{FianlGradient}. Using the inequality $\mathsf{E}[\boldsymbol{\Delta}_k\boldsymbol{\Delta}_k^H] \le \kappa^2 \mathbf{I}$ and $\mathsf{E}[\boldsymbol{\Delta}_k^H\boldsymbol{\Delta}_j] = \frac{1}{K-1}[\boldsymbol{\Delta}_k^H\sum\limits_{j \ne k}\boldsymbol{\Delta}_j] = \frac{1}{K-1}\boldsymbol{\Delta}_k^H (\sum\limits_{j}\boldsymbol{\Delta}_j-\boldsymbol{\Delta}_k) = -\frac{1}{K-1}\boldsymbol{\Delta}_k^H\boldsymbol{\Delta}_k < 0$, we further have 
	\begin{align}\label{signalStep}
	&\mathsf{E}\left[F(\mathbf{W}_{n})-F(\mathbf{W}_{n+1})\right]\nonumber\\ \ge &\frac{\mu}{2}\bigg(\| \vec{g}(\mathbf{W}_{n})\|^2 - \mu \beta \sum\limits_{i=1}^c\mathsf{E}\bigg[\frac{\kappa^2}{K_n^{(i)}}+\frac{\nu^2 \sigma^2}{(K_n^{(i)})^2 P^{\text{rx}}_n} \bigg]\bigg).
	\end{align}
	The number of active devices $K_n^{(i)}$ follows a binomial distribution $K_n^{(i)} \sim B(K,{\zeta^\text{act}})$ with ${\zeta^\text{act}}$ denoting the activation probability in \eqref{nonTruncation}. This leads to the following results:
	\begin{align}
	\Pr\big(K_n^{(i)} = k\big) = \dbinom{K}{k}{{\zeta^\text{act}} ^k}{(1 - {\zeta^\text{act}} )^{K - k}}.
	\end{align}
	It follows that
	\begin{align}\label{NoiseSimple}
	\mathsf{E}\bigg[\!\frac{\kappa^2}{K_n^{(i)}}\!+\!\frac{\sigma^2}{(K_n^{(i)})^2}\!\bigg] \!=\! \sum\limits_{k=1}^K\! \dbinom{K}{k}{{\zeta^\text{act}} ^k}{(1 \!-\! {\zeta^\text{act}} )^{K \!- k}}\!  \left[\!\frac{\kappa^2}{k}\!+\!\frac{\sigma^2}{k^2}\right].
	\end{align}
	Then based on the inequalities established in Lemma \ref{inequality} in Sec. \ref{sec:lem1}, we have
	\begin{align}\label{inequal}
	\mathsf{E}\bigg[ \sum\limits_{i=1}^c\bigg(\frac{\kappa^2}{K_n^{(i)}}+\frac{\sigma^2}{(K_n^{(i)})^2} \bigg)\bigg] \le c\left[\frac{2\kappa^2}{k}+\frac{6 \sigma^2}{k^2}\right].
	\end{align}
	By substituting \eqref{inequal} into \eqref{signalStep}, and considering that $\|\vec{g}(\mathbf{W}_n)\|^2 \ge \epsilon^2$ in $\mathcal{R}_\text{ns}$, we further have
	\begin{align}
	\mathsf{E}\left[F(\mathbf{W}_{\!n})\!-\!F(\mathbf{W}_{\!n+1})\right]
	\!\ge\!  {\mu^2}   \!\!\left[\!\frac{\epsilon^2}{2\mu} \!-\! \beta c\!\left(\! \frac{ \kappa^2}{K \!{\zeta^\text{act}}} \!+\! \frac{3\nu^2\sigma^2}{K^2{\zeta^\text{act}}^2 \! P^{\text{rx}}_{n}} \!\right) \right],
	\end{align}
	which directly gives \eqref{CompareRegion1} by telescoping over the first $n$ steps. This completes the proof.

	\subsection{Lemma 1 and Its Proof}\label{sec:lem1}
	\begin{lem}\label{inequality}
		The following two important inequalities hold:
		\begin{align}
		f(K,{\zeta^\text{act}}) &\equiv \sum\limits_{k=1}^K\ \frac{1}{k} \dbinom{K}{k}{{\zeta^\text{act}} ^k}{(1 - {\zeta^\text{act}} )^{K - k}} \le \frac{2}{K\!{\zeta^\text{act}}}.\nonumber\\
		h(K,{\zeta^\text{act}}) &\equiv \sum\limits_{k=1}^K \frac{1}{k^2} \dbinom{K}{k}{{\zeta^\text{act}} ^k}{(1 - {\zeta^\text{act}} )^{K - k}} \le \frac{6}{K^2{\zeta^\text{act}}^2}.
		\end{align} 
	\end{lem}
	
	\begin{proof}
		Here we prove only the second inequality as the first one can be derived in the same way.
		Function $h(K, {\zeta^\text{act}})$ can be rewritten as
		\begin{align}
		&h(K, {\zeta^\text{act}})\nonumber\\
		= &K{\zeta^\text{act}}  \sum\limits_{k=1}^K \frac{1}{k^3} \binom{K-1}{k-1}{{\zeta^\text{act}}^{k-1}}{(1 - {\zeta^\text{act}} )^{K-k-1}},\nonumber\\
		= &K{\zeta^\text{act}}  \sum\limits_{k=0}^{K-1} \frac{1}{(k+1)^3} \binom{K-1}{k}{{\zeta^\text{act}}^{k}}{(1 - {\zeta^\text{act}} )^{K-k-1}},\nonumber\\
		= &K{\zeta^\text{act}}  \sum\limits_{k=0}^{K-1} \frac{k+3}{k+1}\cdot\frac{k+2}{k+1}\cdot\frac{1}{(k+1)(k+2)(k+3)}\nonumber\\ 
		&\qquad\qquad \cdot \binom{K-1}{k}{{\zeta^\text{act}}^{k}}{(1 - {\zeta^\text{act}} )^{K-k-1}}.\nonumber
		\end{align}
		Since $\frac{k+3}{k+1} \le 3$ and $\frac{k+2}{k+1} \le 2$, the function $h(K, {\zeta^\text{act}})$ can be bounded by
		\begin{align}
		h(K, {\zeta^\text{act}}) \le &\frac{6K{\zeta^\text{act}}}{K(K+1)(K+2){\zeta^\text{act}}^3}\nonumber\\
		&\times \sum\limits_{k=0}^{K-1}\binom{K+2}{k+3}{{\zeta^\text{act}}^{k+3}}{(1 - {\zeta^\text{act}} )^{K-k-1}},\nonumber\\
		<&\frac{6K{\zeta^\text{act}}}{K(K+1)(K+2){\zeta^\text{act}}^3}\nonumber\\
		&\times \underbrace{\sum\limits_{k=-3}^{K-1}\binom{K+2}{k+3}{{\zeta^\text{act}}^{k+3}}{(1 - {\zeta^\text{act}} )^{K-k-1}}}_{=1}, \nonumber\\
		=&\frac{6}{(K+1)(K+2){\zeta^\text{act}}^2}< \frac{6}{K^2{\zeta^\text{act}}^2}.
		\end{align}
		This finishes the proof.
	\end{proof}

	\subsection{Proof of Theorem \ref{saddle}}\label{B}
	Given that $\mathbf{W}_n \in \mathcal{R}_\text{sa}$, and according to the $\chi$-Lipschitz Hessian in \eqref{Hessian}, we have
	\begin{align*}
	&F(\mathbf{W}_{n+1})\!-\!F(\mathbf{W}_{n})\nonumber\\
	\le&\vec{g}(\mathbf{W}_{\!n})^{\!T}(\vec{\mathbf{w}}_{n+1}\!\!-\!\vec{\mathbf{w}}_n)\!+\!\frac{1}{2}(\vec{\mathbf{w}}_{n+1}\!\!-\!\vec{\mathbf{w}}_n)^{\!T}\mathcal{H}(\mathbf{W}_{\!n})(\vec{\mathbf{w}}_{n+1}\!\!-\!\vec{\mathbf{w}}_n)\nonumber\\
	&+\frac{\chi}{6}\|\vec{\mathbf{w}}_{n+1}\!-\!\vec{\mathbf{w}}_n\|^3,\nonumber\\
	&F(\mathbf{W}_{n+1})\!-\!F(\mathbf{W}_{n})\nonumber\\
	\ge&\vec{g}(\mathbf{W}_{\!n})^{\!T}(\vec{\mathbf{w}}_{n+1}\!\!-\!\vec{\mathbf{w}}_n)\!+\!\frac{1}{2}(\vec{\mathbf{w}}_{n+1}\!-\!\vec{\mathbf{w}}_n)^{\!T}\mathcal{H}(\mathbf{W}_{\!n})(\vec{\mathbf{w}}_{n+1}\!\!-\!\vec{\mathbf{w}}_n)\nonumber\\
	&-\frac{\chi}{6}\|\vec{\mathbf{w}}_{n+1}\!-\!\vec{\mathbf{w}}_n\|^3,
	\end{align*}
	where $\vec{g}(\mathbf{W}_n)$ is the real vectorized gradient and $\vec{\mathbf{w}}$ is the real vectorization of matrix $\mathbf{W}$. Note that the last term is ${O}(\mu^3\|\vec{\boldsymbol{\xi}}_n\|^3)$, which is negligible compared to the first two terms with sufficiently small step-size $\mu \ll 1/c \left[ \frac{\kappa^2}{K{\zeta^\text{act}}}+\frac{3\nu^2  \sigma^2}{K^2{\zeta^\text{act}}^2 P^\text{rx}_{\min}} \right]$. In this case, we have
	\begin{align}\label{taylor}
	&F(\mathbf{W}_{n+1})\!-\!F(\mathbf{W}_{n}) \nonumber\\
	\rightarrow&\vec{g}(\mathbf{W}_{\!n})^{\!T}\!(\vec{\mathbf{w}}_{n+1}\!\!-\!\vec{\mathbf{w}}_n)\!+\!\frac{1}{2}(\vec{\mathbf{w}}_{n+1}\!\!-\!\vec{\mathbf{w}}_n)^{\!T}\!\mathcal{H}(\mathbf{W}_{\!n})(\vec{\mathbf{w}}_{n+1}\!\!-\!\vec{\mathbf{w}}_n),
	\end{align}
	which means we can treat $F(\mathbf{W})$ as a locally quadratic function with negligible deviation.
	Denote $\boldsymbol{\mathcal{H}} = \mathcal{H}(\mathbf{W}_0)$ as the Hessian matrix at $\mathbf{W}_0$. It follows that
	\begin{align}\label{GraTaylor}
	\vec{g}(\mathbf{W}_n) \rightarrow &\vec{g}(\mathbf{W}_0) + \boldsymbol{\mathcal{H}}(\mathbf{W}_n-\mathbf{W}_0) \nonumber\\
	= &(\mathbf{I}\!-\!\mu\boldsymbol{\mathcal{H}})^n \vec{g}(\mathbf{W}_0)\!-\!\mu\boldsymbol{\mathcal{H}}\sum\limits_{{m}=0}^{n-1}(\mathbf{I}\!-\!\mu\boldsymbol{\mathcal{H}})^{n-{m}-1}\vec{\boldsymbol{\xi}}_{m},
	\end{align}
	and
	\begin{align}\label{DisTaylor}
	&\vec{\mathbf{w}}_n \!-\! \vec{\mathbf{w}}_0\nonumber\\
	=&\!-\mu\!\sum\limits_{{m}=0}^{n-1}\!\left(\vec{g}(\mathbf{W}_{m})\!+\!\vec{\boldsymbol{\xi}}_n \right),\nonumber\\
	\rightarrow&\!-\!\mu\!\!\sum\limits_{{m}=0}^{n-1}\!\!\!\bigg(\!\!(\mathbf{I}\!\!-\!\!\mu\boldsymbol{\mathcal{H}})^{m} \vec{g}(\mathbf{W}_0)\!-\!\mu\boldsymbol{\mathcal{H}}\!\!\!\sum\limits_{{m}^\prime = 0}^{{m}\!-\!1}\!\!(\mathbf{I}\!\!-\!\!\mu\boldsymbol{\mathcal{H}})^{{m} \!-\! {m}^\prime\!-\!1}\vec{\boldsymbol{\xi}}_{{m}^\prime} \!\!+\! \vec{\boldsymbol{\xi}}_{m} \!\!\bigg),\nonumber\\
	=& -\mu\sum\limits_{{m}=0}^{n-1}(\mathbf{I}-\mu\boldsymbol{\mathcal{H}})^{m} \vec{g}(\mathbf{W}_0)-\mu\sum\limits_{{m}=0}^{n-1}(\mathbf{I}-\mu\boldsymbol{\mathcal{H}})^{t-{m}-1}\vec{\boldsymbol{\xi}}_{m}.
	\end{align}
	
	Combining \eqref{taylor}, \eqref{GraTaylor} and \eqref{DisTaylor} gives
	\begin{align}
	&\mathsf{E}[F(\mathbf{W}_{0})- F(\mathbf{W}_{n})]\nonumber\\
	=&\frac{\mu}{4}\!\sum\limits_{i=1}^{2c-1}\!\sum\limits_{{m}=0}^{2n-1}\!(1\!-\!\mu\lambda_i)^{m}|[\vec{g}(\mathbf{W}_0)]_i|^2 \nonumber\\
	&\!-\!\!\frac{\mu^2}{4}\!\!\sum\limits_{i=1}^{2c-1}\!\!\lambda_i\!\! \sum\limits_{{m}=0}^{n-1}\!(1\!\!-\!\!\mu\lambda_i)^{\!2(t\!-\!{m}\!-\!1)}\mathsf{E}\!\left[\!{[\vec{\boldsymbol{\xi}}_{m}]_i}^{\!\!\!T}\! [\vec{\boldsymbol{\xi}}_{m}]_i\!\right]\!\! -\!\! \sum\limits_{{m}=0}^{n-1}\!\!\mu^3\|\vec{\boldsymbol{\xi}}_{m}\|^3,\nonumber\\
	\label{increasing}
	\ge& \!-\!\!\frac{1}{4}\!\!\sum\limits_{i=1}^{2c-1}\!\!\lambda_i\!\! \sum\limits_{{m}=0}^{n-1}\!(1\!-\!\mu\lambda_i)^{2(n\!-\!{m}\!-\!1)}\!\mu^2\mathcal{V}_{m}^{(i)}\!-\!\!\sum\limits_{{m}=0}^{n-1}\!\mu^3\|\vec{\boldsymbol{\xi}}_{m}\|^3, \\ 
	\label{R2proof}
	\ge& \frac{1}{4}  \sum\limits_{{m}=0}^{n-1}\mu^2\gamma(1\!+\!\mu\gamma)^{2(n-{m}-1)}\frac{\nu^2 \sigma^2}{K^2 P^\text{rx}_m}\nonumber\\
	&\!-\!\frac{\mu(2c\!-\!1)}{4}\!\left[\!\frac{\kappa^2}{K\zeta^\text{act}}\!+\!\frac{3\nu^2  \sigma^2}{K^2{\zeta^\text{act}}^2 P^{\text{rx}}_{\min}} \right] \!\!-\!\! \sum\limits_{{m}=0}^{n-1}\!\!\mu^3\|\vec{\boldsymbol{\xi}}_{m}\|^3,
	\end{align}
	where $\mathcal{V}_{m}^{(i)} = \mathsf{E}\left[{[\vec{\boldsymbol{\xi}}_{m}]_i}^T [\vec{\boldsymbol{\xi}}_{m}]_i\right]$ is the noise variance on the $i$-th element in the received gradient. Note that in \eqref{increasing} we have $\mu\lambda_i \sum\limits_{{m}=0}^{n-1}(1-\mu\lambda_i)^{2(n-{m}-1)} \le 1$ for $\lambda_i \ge 0$, and $\sum\limits_{{m}=0}^{n-1}(1-\mu\lambda_i)^{2(n-{m}-1)}$ is monotonically increasing w.r.t. $n$ for $\lambda_i < 0$. Then the last inequality can be obtained by using $2\mathsf{E}\left[{[\vec{\boldsymbol{\xi}}_{m}]_i}^T [\vec{\boldsymbol{\xi}}_{m}]_i\right] \le \frac{\kappa^2}{K\zeta^\text{act}}+\frac{3\nu^2  \sigma^2}{K^2{\zeta^\text{act}}^2 P^{\text{rx}}_{\min}}$ and $2\mathsf{E}\left[{[\vec{\boldsymbol{\xi}}_{m}]_i}^T [\vec{\boldsymbol{\xi}}_{m}]_i\right] \ge \frac{\nu^2 \sigma^2}{K^2 P^\text{rx}_m}$.	Furthermore, by setting the two constants $\mathcal{V}_{\max} =  \frac{\kappa^2}{K\zeta^\text{act}}+\frac{3\nu^2  \sigma^2}{K^2{\zeta^\text{act}}^2 P^{\text{rx}}_{\min}}$, $\mathcal{V}_{\min} =  \frac{\kappa^2}{K}+\frac{\nu^2  \sigma^2}{K^2 P^{\text{rx}}_{\max}}\le \frac{\nu^2 \sigma^2}{K^2 P^{\text{rx}}_{m}}$, \eqref{R2proof} can be further bounded by
	\begin{align}\label{Expectation_L2}
	&\mathsf{E}[F(\mathbf{W}_{0})- F(\mathbf{W}_{n})] \nonumber\\
	\ge& \frac{1}{4}  \sum\limits_{{m}=0}^{n-1}\mu^2\gamma(1+\mu\gamma)^{2(n-{m}-1)} \frac{\nu^2 \sigma^2}{K^2 P^\text{rx}_m}\nonumber\\
	&-\underbrace{\frac{\mu(2c-1)}{4}\left[\frac{\kappa^2}{K\zeta^\text{act}}+\frac{3\nu^2  \sigma^2}{K^2{\zeta^\text{act}}^2 P^{\text{rx}}_{\min}} \right] - \sum\limits_{{m}=0}^{n-1}\mu^3\|\vec{\boldsymbol{\xi}}_{m}\|^3}_{O(\mu)}, \nonumber\\
	\ge&\!-\!\frac{\mu\mathcal{V}_{\min}}{4}\!\bigg(\!(2c\!-\!1)\frac{\mathcal{V}_{\max}}{\mathcal{V}_{\min}} \mu\gamma\!\sum\limits_{{m}=0}^{n-1}\!(1\!+\!\mu\gamma)^{2{m}} \bigg) \!-\! n\!\cdot\!{O}(\mu^3).
	\end{align}
	With step-size $\mu \ll 1/c \left[ \frac{\kappa^2}{K{\zeta^\text{act}}}+\frac{3\nu^2  \sigma^2}{K^2{\zeta^\text{act}}^2 P^{\text{rx}}_{\min}} \right]$, the second term $n\cdot O(\mu^3)$ is negligible to other terms. Then according to \eqref{Expectation_L2}, we look for a $N_{\max}$ that enables $2c\frac{\mathcal{V}_{\max}}{\mathcal{V}_{\min}} \! \le\! \mu\gamma\!\!\sum\limits_{{m}=0}^{N_{\max}\!-1}\!(1\!+\!\mu\gamma)^{2{m}}$, where a sufficient condition is
	\begin{align*}
	&2c\frac{\mathcal{V}_{\max}}{\mathcal{V}_{\min}} \!\le\! \frac{(1\!\!+\!\!\mu\gamma)^{2N_{\max}}\!-\!1}{3}\!\!\iff\!\! 6c\frac{\mathcal{V}_{\max}}{\mathcal{V}_{\min}}\!\!+\!\!1 \!\le\! (1\!\!+\!\!\mu\gamma)^{\!\frac{2\mu\gamma N_{\max}}{\mu\gamma}}\nonumber\\
	&\iff\! N_{\max} \!\ge\! \frac{\log (6c\frac{\mathcal{V}_{\max}}{\mathcal{V}_{\min}}\!+\!1)}{2\mu\gamma}.
	\end{align*}
	Therefore, by taking $N_{\max} = \frac{\log (6c\frac{\mathcal{V}_{\max}}{\mathcal{V}_{\min}}+1)}{2\mu\gamma}$, we can simplify \eqref{Expectation_L2} as
	\begin{align}\label{L2guarantee}
	&\mathsf{E}[F(\mathbf{W}_{0})-F(\mathbf{W}_{N_{\max}})]\nonumber\\
	\ge&\!-\!\frac{\mu\mathcal{V}_{\min}}{4}\!\bigg(\!(2c\!-\!1)\frac{\mathcal{V}_{\max}}{\mathcal{V}_{\min}} \!-\!\mu\gamma\!\!\!\!\sum\limits_{{m}=0}^{N_{\max}\!-\!1}\!\!\!(1\!+\!\mu\gamma)^{2{m}} \!\bigg)\!\!-\!\! N_{\max}\!\cdot\!{O}(\mu^3), \nonumber\\
	\ge& \frac{\mu\mathcal{V}_{\max}}{4}- {O}(\mu^2) \rightarrow \frac{\mu\mathcal{V}_{\max}}{4}.
	\end{align}
	Combining \eqref{L2guarantee} and \eqref{R2proof} gives
	\begin{align}
	&\mathsf{E}[F(\mathbf{W}_{0})- F(\mathbf{W}_{n})]\nonumber\\ 
	\ge &\frac{\mu^2\gamma}{4}   \sum\limits_{m=0}^{n\!-\!N_{\max}\!-\!1}(1+\mu\gamma)^{2(n-m-1)}\frac{\nu^2 \sigma^2}{K^2 P^\text{rx}_m} +\frac{\mu\mathcal{V}_{\max}}{4},
	\end{align}
	for $n \ge N_{\max}$. The result in \eqref{L2_approx} directly follows.   
	Moreover, since $N_{\max} = \frac{\log (6c\frac{\mathcal{V}_{\max}}{\mathcal{V}_{\min}}+1)}{2\mu\gamma}$, then with unbiased noise bounded by $\|\vec{\boldsymbol{\xi}}_n\|\le {O}(1)$ with probability $1$. According to the Hoeffding inequality,
	\begin{align}
	\!\!\!\Pr\!\bigg(\!\!\|\mu\!\!\sum\limits_{{m}=0}^{n-1}\!(1\!\!-\!\!\mu\lambda_i)^{n\!-\!{m}\!-\!1}[\vec{\boldsymbol{\xi}}_{{m}}]_i\|\!\!>\!16c\frac{\mathcal{V}_{\max}}{\mathcal{V}_{\min}}O\!\left( \!\!\sqrt{\!\mu\log \!\frac{1}{\mu}}\right)\!\!\!  \bigg) \!\!&\le\! \mu^2, \nonumber\\
	\forall n \le N_{\max}. &
	\end{align}
	By summing over dimension $c$ and taking union bound over all $n\le N_{\max}$, it follows that
	\begin{align}\label{stable}
	&\Pr \!\bigg(\!\forall n\!\le\! N_{\max}, \|\mu\!\!\sum\limits_{{m}=0}^{n-1}\!(\mathbf{I}\!-\!\mu\boldsymbol{\mathcal{H}})^{n\!-\!{m}\!-\!1}\vec{\boldsymbol{\xi}}_{m}\|\!>\!{O}\bigg(\!\!\sqrt{\mu\log \!\frac{1}{\mu}}\bigg) \!\!\bigg) \nonumber\\
	\le &{O}(\mu).
	\end{align}
	Note that $\mu$ can be chosen such that $16c\frac{\mathcal{V}_{\max}}{\mathcal{V}_{\min}}\!\sqrt{\!\mu\log \frac{1}{\mu}} \!\le\! \epsilon$. Then combine \eqref{stable} with \eqref{GraTaylor} and \eqref{DisTaylor},
	\begin{align}\label{DisNmax}
	\|\mathbf{W}_n -\mathbf{W}_0\|\le {O}(\mu^{\frac{1}{2}}\log \frac{1}{\mu})\le \epsilon, \quad \forall n\le N_{\max},\nonumber\\
	\|\vec{g}(\mathbf{W}_n) -\vec{g}(\mathbf{W}_0)\|\le {O}(\mu^{\frac{1}{2}}\log \frac{1}{\mu}), \quad \forall n\le N_{\max},
	\end{align}
	with probability as least $1-{O}(\mu)$. The results in \eqref{DisNmax} indicates that the distance $\|\mathbf{W}_n-\mathbf{W}_0\|$
	keeps to be small in the $n$-round process, which also validates the Taylor-approximation here.

	\subsection{Proof of Theorem \ref{SurelyConverge}}\label{D}
	Firstly, consider the descent in the non-stationary region. According to \eqref{signalStep}, given a step-size $\mu \le \frac{\epsilon^2}{4\beta c} \left[ \frac{ \kappa^2}{K {\zeta^\text{act}}} + \frac{3\nu^2 \sigma^2}{K^2{\zeta^\text{act}}^2 P^{\text{rx}}_{\min}} \right]^{-1}$,
	\begin{align}
	&\mathsf{E}[F(\mathbf{W}_{n})-F(\mathbf{W}_{n+1})|\mathbf{W}_n \in \mathcal{R}_\text{ns}] \nonumber\\
	\ge&{{\mu^2\beta} \sum\limits_{i=1}^c\left(\frac{\kappa^2}{K\zeta^\text{act}}+\frac{3\nu^2 \sigma^2}{K^2{\zeta^\text{act}}^2 P^{\text{rx}}_{\min}} \right)} = \mu^2\beta c \mathcal{V}_{\max}.
	\end{align}
	It follows that
	\begin{align}\label{descent_R1}
	&\mathsf{E}[F(\mathbf{W}_{n})-F(\mathbf{W}_{n+N_{\max}})|\mathbf{W}_n \in \mathcal{R}_\text{ns}] \nonumber\\
    \ge&\frac{\mu}{2} \mathcal{V}_{\max} \left[\frac{\beta c}{\gamma} \log(6c\frac{\mathcal{V}_{\max}}{\mathcal{V}_{\min}}+1) \right],
	\end{align}
	where $N_{\max} = \frac{\log (6c\frac{\mathcal{V}_{\max}}{\mathcal{V}_{\min}}+1)}{2\mu\gamma}$ is a constant defined in Theorem \ref{saddle}.
	Secondly, consider the descent process in the a saddle region. According to Theorem \ref{saddle}, with a step-size satisfying \eqref{stepsizeL2},
	\begin{align}\label{descent_R2}
	\mathsf{E}[F(\mathbf{W}_{n}) - F(\mathbf{W}_{n+N_{\max}})|\mathbf{W}_n \in \mathcal{R}_\text{sa}] \ge \frac{\mu}{4}\mathcal{V}_{\max}.
	\end{align}
	By combining \eqref{descent_R1} and \eqref{descent_R1},
	\begin{align}
	\mathsf{E}[F(\mathbf{W}_{n}) - F(\mathbf{W}_{n+N_{\max}})|\mathbf{W}_n \notin \mathcal{R}_\text{op}] \ge \frac{\mu \rho}{4}\mathcal{V}_{\max},   
	\end{align}
	where $\rho = \min\{\frac{2\beta c}{\gamma} \log(6c\frac{\mathcal{V}_{\max}}{\mathcal{V}_{\min}}+1), 1 \}$ is a constant.	
	Define event $\mathcal{E}_n = \{\exists j\le n, \mathbf{W}_n \in \mathcal{R}_\text{op} \}$, clearly $\mathcal{E}_n \subset \mathcal{E}_{n+N_{\max}}$, thus $\Pr(\mathcal{E}_n) \le \Pr(\mathcal{E}_{n+N_{\max}})$. Finally, consider $F(\mathbf{W}_{n+N_{\max}})\mathsf{1}_{\mathcal{E}_n}$, where $\mathsf{1}_{\mathcal{E}_n}$ is an indicator function which is $1$ when event $\mathcal{E}_n$ is true and $0$ otherwise. Then we have
	\begin{align}\label{prob}
	&\mathsf{E}[F(\mathbf{W}_{n+N_{\max}})\mathsf{1}_{\mathcal{E}_n}] - \mathsf{E}[F(\mathbf{W}_{n})\mathsf{1}_{\mathcal{E}_{n-N_{\max}}}]\nonumber\\
	 \le& B \cdot [\Pr(\mathcal{E}_n)-\Pr(\mathcal{E}_{n-N_{\max}})]\nonumber\\
	 &+ \mathsf{E}[F(\mathbf{W}_{n+N_{\max}}) - F(\mathbf{W}_{n}) \big| \bar{\mathcal{E}}_n]\cdot \Pr(\bar{\mathcal{E}}_n),\nonumber\\
	\le& B \cdot [\Pr(\mathcal{E}_n)-\Pr(\mathcal{E}_{n-N_{\max}})] - \frac{\mu \rho}{4}\mathcal{V}_{\max}\cdot \Pr(\bar{\mathcal{E}}_n),
	\end{align}
	where $B$ is the upper-bound on the error function norm, given as $\|F(\mathbf{W})\|\le B$ for all $\mathbf{W}$. The term $\Pr(\mathcal{E}_n)-\Pr(\mathcal{E}_{n-N_{\max}})$ denotes the probability that the descent process enters $\mathcal{R}_\text{op}$ for the first time between the $(n-N_{\max})$-th  and the $n$-th round, and $\Pr(\bar{\mathcal{E}}_n)$ denotes the probability that the descent process never enters into $\mathcal{R}_\text{op}$ in the first $n$ rounds. Therefore, by summing up \eqref{prob} over a long period $m \cdot N_{\max}$ with $m\in\mathds{N}^+$,
	\begin{align}
	&\mathsf{E}[F(\mathbf{W}_{mN_{\max}})\mathsf{1}_{\mathcal{E}_{(m-1)N_{\max}}}] - F(\mathbf{W}_{0})\nonumber\\ \le& B \cdot \Pr(\mathcal{E}_{(m-1)N_{\max}}) - \frac{ \mu \rho}{4}\mathcal{V}_{\max}\cdot \sum\limits_{i=1}^m\Pr(\bar{\mathcal{E}}_{(i-1)N_{\max}}),\nonumber\\
	\le& B - \frac{m \mu \rho}{4}\mathcal{V}_{\max}\cdot \Pr(\bar{\mathcal{E}}_{(m-1)N_{\max}}).
	\end{align}
	Since $\|F(\mathbf{W}_{mN_{\max}})\mathsf{1}_{\mathcal{E}_{(m-1)N_{\max}}}\| \le B$ is bounded, we have 
	\begin{align}
	\Pr(\bar{\mathcal{E}}_{(m-1)N_{\max}}) \le \frac{12B}{m \mu \rho \mathcal{V}_{\max}},
	\end{align}
	which gives Theorem \ref{SurelyConverge}.
	Based on the above conclusion, it is also obvious that the process enters $\mathcal{R}_\text{op}$ at least once with probability $1$ when $m \to \infty$, i.e., $n \to \infty$. This finishes the proof.

	\bibliographystyle{ieeetr}
	\bibliography{Ref}

\end{document}